\pdfoutput=1
\RequirePackage{amsmath}
\documentclass[a4paper,12pt,bold]{scrartcl}

\normalsize

\newcommand{\thin}{\thinspace}

\newcommand{\E}{\mathrm{E}}		
\newcommand{\R}{\mathbb{R}}

\newcommand{\A}{\mathcal{A}}
\newcommand{\Pc}{\mathcal{P}}
\newcommand{\ST}{\mathcal{S}}		
\newcommand{\C}{\mathcal{C}}		

\usepackage{bbm}
\renewcommand{\vec}[1]{\mathbf{#1}}

\usepackage{algpseudocode,tabularx,ragged2e}
\newcolumntype{C}{>{\centering\arraybackslash}X} 
\newcolumntype{L}{>{\arraybackslash}X} 

\usepackage{algorithmicx}
\usepackage{graphicx}
\usepackage{afterpage}
\usepackage{algorithm}
\usepackage{subfig}
\usepackage{float}
\usepackage[section]{placeins}
\usepackage{apacite}
\usepackage{booktabs}
\usepackage{epigraph}
\usepackage[sans]{dsfont}
\usepackage[round,longnamesfirst]{natbib}
\usepackage{bm}																									
\usepackage{setspace}																						
\usepackage[colorlinks = true,
            linkcolor = blue,
            urlcolor  = blue,
            citecolor = blue,
            anchorcolor = blue]{hyperref}
\usepackage{threeparttable}
\usepackage{lscape}																							
\usepackage[latin1]{inputenc}
\usepackage{graphicx}
\usepackage{amsmath}
\usepackage{amsthm}
\usepackage{mathtools}
\usepackage{amssymb}
\usepackage{bigints}
\usepackage{fancybox}																						
\usepackage{appendix}
\usepackage{listings}
\usepackage{xr}
\usepackage{pdflscape}

\usepackage{enumerate}
\usepackage[labelfont=bf]{caption}
\usepackage{tabularx}
\usepackage{longtable}
\usepackage{float}																				
\usepackage{color,colortbl}																			
\usepackage[left=2cm, right=2cm, top=2cm, bottom=2.5cm]{geometry} 

\definecolor{lightgrey}{gray}{0.95}	
\definecolor{grey}{gray}{0.85}
\definecolor{darkgrey}{gray}{0.80}

\usepackage{rotating}

\usepackage{tikz}
\usetikzlibrary{calc,trees,positioning,arrows,fit,shapes,calc}
\usepackage[export]{adjustbox}
\usepackage{authblk}
\usepackage{caption}

\usepackage{url}  
\usepackage{multirow}

\newtheorem{Theorem}{Theorem}
\newtheorem{Assumption}{Assumption}

\deffootnote[1em]{1.1em}{0em}{\textsuperscript{\thefootnotemark}}
\renewcommand{\arraystretch}{1.05}

\DeclareMathOperator*{\argmin}{arg\,min}
\DeclareMathOperator*{\argmax}{arg\,max}

\makeatletter

\defcitealias{Rust.1987}{Rust's (1987)}

\setkomafont{author}{\scshape}
\usepackage{blindtext}
\numberwithin{equation}{section}
\graphicspath{{material/}{material/decision-theory-overview/}}

\title{Robust decision-making under risk and ambiguity\thanks{Corresponding author: Philipp Eisenhauer, peisenha@uni-bonn.de. Philipp Eisenhauer is funded by the Deutsche Forschungsgemeinschaft (DFG, German Research Foundation) under Germany's Excellence Strategy - EXC 2126/1- 390838866, the TRA Modelling (University of Bonn) as part of the Excellence Strategy of the federal and state governments, and a postdoctoral fellowship by the AXA Research Fund. We thank Anton Bovier, Annica Gehlen, Lena Janys, Ken Judd, John Kennan, Simon Merkt, Sven Rady, Gregor Reich, John Rust, J{\"o}rg Stoye, and Rafael Suchy for numerous helpful discussions. We thank Annica Gehlen for her outstanding research assistance. We are grateful to the Social Sciences Computing Service (SSCS) at the University of Chicago for the permission to use their computational resources. We gratefully acknowledge support by the AXA Research Fund.}}

\author[1]{Maximilian Blesch}
\author[2]{Philipp Eisenhauer}
\affil[1]{Berlin School of Economics}
\affil[2]{University of Bonn}

\date{\today}

\begin{document}

\maketitle

\vspace{0.5cm}\begin{abstract}\noindent
Economists often estimate economic models on data and use the point estimates as a stand-in for the truth when studying the model's implications for optimal decision-making. This practice ignores model ambiguity, exposes the decision problem to misspecification, and ultimately leads to post-decision disappointment. Using statistical decision theory, we develop a framework to explore, evaluate, and optimize robust decision rules that explicitly account for estimation uncertainty. We show how to operationalize our analysis by studying robust decisions in a stochastic dynamic investment model in which a decision-maker directly accounts for uncertainty in the model's transition dynamics.\\
\end{abstract}

\noindent\begin{tabular}{@{\hspace{0.5cm}}ll}
\textbf{JEL Codes} & D81, C44, D25\\
\textbf{Keywords}  & decision-making under uncertainty, robust Markov decision process
\end{tabular}

\setcounter{page}{1}
\thispagestyle{empty}

\newpage
\tableofcontents
\newpage

\FloatBarrier\section{Introduction}
Decision-makers often confront uncertainties when determining their course of action. For example, individuals save to cover uncertain medical expenses in old age \citep{French.2014}. Firms set prices in an uncertain competitive environment \citep{Ilut.2020}, and policy-makers face uncertainties about future costs and benefits when voting on climate change mitigation efforts \citep{Barnett.2020}. We consider the situation in which a decision-maker posits a collection of economic models to inform his decision-making process. Each model formalizes the relevant objectives and trade-offs involved and provides an implicit rule for optimal decisions. Uncertainty is limited to risk for a given model, as the model induces a unique probability distribution over possible future outcomes. However, a decision-maker also faces model ambiguity as the true model within the collection remains uncertain \citep{Arrow.1951,Knight.1921}.\\

It is the standard practice in economics to estimate models on data and use the point estimates as a stand-in for the truth when studying the model's implications and optimal decision-making.\footnote{See examples in labor economics \citep{Adda.2017,Blundell.2012}, industrial organization \citep{Hortacsu.2019,Igami.2017}, and international trade \citep{Bagwell.2021,Eaton.2011}.} This approach ignores model ambiguity, resulting from the remaining parametric uncertainty after the estimation, and opens the door for the misspecification of the decision problem. As-if decisions, decisions that are optimal if the point estimates used to inform decisions are correct \citep{Manski.2021}, often turn out to be very sensitive to misspecification \citep{Smith.2006}. This danger creates the need for robust decisions that perform well over a whole range of different models instead of as-if decisions that perform best for one particular model.  However, increasing the robustness of decisions, often measured by a performance guarantee under a worst-case scenario, reduces performance in all other cases. Striking a balance between the two objectives is challenging.\\

We solve this trade-off and determine the optimal level of robustness by combining insights from statistical decision theory \citep{Berger.2010} with data-driven robust optimization \citep{Bertsimas.2018}. A core concept in statistical decision theory is a statistical decision function (SDF) that provides a procedure to map all available data into decisions. At the same time, the literature on data-driven robust optimization provides us with precisely such procedures for decision-making with varying levels of robustness against misspecification of the decision problem. Our main contribution is interpreting these procedures as SDFs and evaluating their performance with the toolkit of statistical decision theory. This insight allows us to systematically determine the optimal level of robustness. In doing so, we bring together and extend research in economics and operations research by using econometric models in complex decision problems \citep{Bertsimas.2006,Manski.2021}.\\

In our application, we revisit \citetalias{Rust.1987} seminal bus replacement problem. Model ambiguity is particularly consequential in dynamic models where the impact of erroneous decisions accumulates over time \citep{Mannor.2007}.\footnote{\citetalias{Rust.1987} model serves as a computational illustration in a variety of settings. See for example \citet{Christensen.2019}, \citet{Iskhakov.2016}, \citet{Reich.2018}, and \citet{Su.2012a}.} In the model, the manager Harold Zurcher implements a maintenance plan for a fleet of buses that maximizes his expected discounted utility. He faces uncertainty about the future mileage utilization of the buses but has data on past utilization available to inform his decisions. While \citetalias{Rust.1987} original goal was to describe the investment behavior of Harold Zurcher, our analysis is normative. We are interested in how a generic decision-maker should make decisions in this instance.\\

The bus replacement problem is typically modeled as a standard Markov decision problem (MDP), and the point estimates for the mileage utilization are treated as-if they correspond to the true parameters. The solution of the MDP is an as-if decision rule that is optimal given the estimates. This approach ignores model ambiguity. From the perspective of statistical decision theory, an MDP is just one particular example of an SDF suitable for analyzing the bus replacement problem. We, on the other hand, consider a whole class of SDFs called robust Markov decision problems (RMDP) \citep{Ben-Tal.2009}. RMDPs generalize the standard MDP, as they consider a whole set of distributions for the transition dynamics collected in an ambiguity set. The solution of an RMDP is a robust decision rule that is optimal under a worst-case scenario for all mileage utilization distributions inside the ambiguity set. We follow the literature and construct the ambiguity set so that it contains all distributions we cannot reject with a certain level of confidence $\omega \in [0, 1]$ around the point estimates under any possible realization of the data \citep{Ben-Tal.2013}. The size of the ambiguity set is a choice by the decision-maker and determines the level of robustness. Given the realization of the data, the robust decision rule based on the solution of an RMDP is always conditional on the specified level of robustness. Each choice of $\omega$ defines a different RMDP, and applying the toolkit of statistical decision theory allows us to determine the optimal level of robustness $\omega^*$ within the whole class of SDFs.\\

To do so, we compare the performance of RMDPs with varying levels of robustness under different decision-theoretic criteria. We consider the situation before any data on mileage utilization is available and implement an ex-ante decision-theoretic analysis. We explore the performance of robust decision rules over the whole probability simplex and are thus able to determine the optimal level of robustness. Throughout, we compare robust and as-if decision rules, as the standard Markov decision problem remains one SDF within the broader class we consider.\\

Figure \ref{RMDPs as statistical decision functions} stresses the point that each RMDP is a different SDF that characterizes robust decisions for any realization of the data. Here, for example, we consider two RMPDs with different levels of robustness -- $\omega_1$ and $\omega_2$ -- that, once data realizes, lead to different decision rules. This situation creates the need to compare their performance under alternative decision-theoretic criteria.
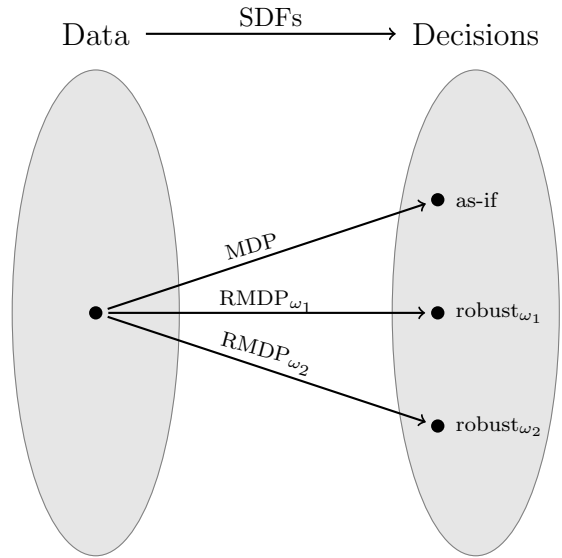
\begin{figure}[h!]\centering
 \centering
 \begin{tikzpicture}[ele/.style={fill=black,circle,minimum width=5pt,inner sep=0.5pt},every fit/.style={ellipse,draw,inner sep=18pt}]

\node (data) at (-1,6.2) {Data};
\node (decisions) at (4,6.2) {Decisions};

\node[] (a1) at (-1,4) {};
\node[] (a2) at (-1,1) {};

\node[] (b1) at (4,4) {};
\node[] (b2) at (4,1) {};

\node[fill,gray!20,fit= (a1) (a2),minimum width=2cm] {} ;
\node[fill,gray!20,fit= (b1) (b2),minimum width=2cm] {} ;

\node[draw, gray,fit= (a1) (a2),minimum width=2cm] {} ;
\node[draw, gray, fit= (b1) (b2),minimum width=2cm] {} ;

\node[ele] (d0) at (-1,2.5) {};

\node[ele, label=right:{\scriptsize as-if}] (d1) at (3.5,4) {};
\node[ele, label=right:{\scriptsize $\text{robust}_{\omega_1}$}] (d2) at (3.5,2.5) {};
\node[ele, label=right:{\scriptsize $\text{robust}_{\omega_2}$}] (d3) at (3.5,1) {};

  \draw[->,thick,shorten <=2pt,shorten >=2pt] (d0) -- node[label={[xshift=-0.2cm, yshift=-0.35cm]\rotatebox{15}{{${\scriptstyle\text{MDP}}$}}}] {}(d1);
  \draw[->,thick,shorten <=2pt,shorten >=2pt] (d0) -- node[label={[xshift=0cm, yshift=-0.3cm]{${\scriptstyle\text{RMDP}_{\omega_1}}$}}] {}(d2);

  \draw[->,thick,shorten <=2pt,shorten >=2pt] (d0) -- node[label={[xshift=0cm, yshift=-0.4cm]\rotatebox{-15}{${\scriptstyle\text{RMDP}_{\omega_2}}$}}] {}(d3);

 \draw[->,thick,shorten <=2pt,shorten >=2pt] (data) -- node[label={[xshift=0cm, yshift=-0.2cm]\footnotesize{SDFs}}] {}(decisions);

 \end{tikzpicture}
\caption{RMDPs as statistical decision functions}\label{RMDPs as statistical decision functions}
\end{figure}\FloatBarrier
Our insight to evaluate robust decisions using statistical decision theory applies to the whole literature on data-driven robust optimization. There exists a growing number of applications of data-driven robust decision-making in a variety of settings, including portfolio decisions \citep{Jin.2020,Zymler.2013}, elective admission to hospitals \citep{He.2019,Meng.2015}, the timing of medical interventions \citep{Goh.2018,Kaufman.2017}, and managing the production of renewable energy \citep{Alismail.2018,Samuelson.2017}.\footnote{See the recent surveys by \citet{Bertsimas.2018}, \citet{Keith.2021}, and \citet{Rahimian.2019} for numerous additional examples.} \\

Despite its broad field of application, the existing literature only offers limited guidance on choosing the optimal level of robustness. At the most basic level, the recommendations range from simply advocating a high level of robustness \citep{Ben-Tal.2013} to choosing a level of robustness that ensures a pre-specified worst-case performance \citep{Brown.2012}. These approaches ignore the trade-off between a performance guarantee under a worst-case scenario and reduced performance in all other cases. Most recently, and much closer to our approach, \citet{Gotoh.2021} put the robustness trade-off front and center. Adopting ideas from the machine learning literature, they calibrate the level of robustness by trading off the mean and variance in the out-of-sample performance of a robust decision rule. However, their approach restricts attention to the neighborhood of a realized point estimate. Thus, their analysis is ex-post as it does not aggregate performance over all possible realizations of the data.\\

At the same time, in econometrics, there is a burgeoning interest in assessing the sensitivity of findings to model or moment misspecification.\footnote{See for example \citet{Andrews.2020}, \citet{Andrews.2017}, \citet{Armstrong.2021}, \citet{Bonhomme.2020}, \cite{Chernozhukov.2020},  \citet{Christensen.2019}, and \citet{Honore.2020}.} Our work is related to \citet{Jorgensen.2021}, who develops a measure to assess the sensitivity of results by fixing a subset of parameters of a model before the estimation of the remaining parameters. Our approach differs as we directly incorporate model ambiguity in the design of the decision-making process and assess the performance of a decision rule under misspecification of the decision environment. As such, our focus on ambiguity faced by decision-makers about the model draws inspiration from the research program summarized in \citet{Hansen.2016} that tackles similar concerns with a theoretical focus. We complement recent work by \cite{Saghafian.2018}, who works in a setting similar to ours, but does not use statistical decision theory to determine the optimal robust decision rule. In ongoing work, \citet{Eisenhauer.2021} use statistical decision theory to structure policy decisions in light of uncertainty about counterfactual policy predictions due to the remaining model ambiguity after the estimation of a model. While they conduct an ex-post evaluation of alternative policy proposals using decision-theoretic criteria, we perform a proper ex-ante analysis of competing decision rules. We evaluate each rule's performance under all possible parameterizations of the model and directly account for the model ambiguity in their construction. In addition, we contribute to the work on optimal treatment allocation started in \citet{Manski.2004} and \citet{Manski.2009}, which characterizes the structure of optimal statistical decision functions and provides (asymptotic) bounds on their performance \citep{Hirano.2009,Stoye.2009,Tetenov.2012,Stoye.2012b,Kitigawa.2018}.\\

The structure of the remaining analysis is as follows. In Section \ref{Framework}, we present statistical decision theory as our framework to compare decision rules. We then set up a canonical model of a data-driven robust Markov decision problem in Section \ref{Conceptual framework} and outline the decision-theoretic determination of the optimal level of robustness. Section \ref{Bus replacement problem} presents our analysis of the robust bus replacement problem. Section \ref{Conclusion} concludes.

\FloatBarrier\section{Statistical decision theory}\label{Framework}
We now show how to compare as-if decision-making to its robust alternatives using statistical decision theory. We first review the basic setting and then turn to a classic urn example to illustrate some key points.

\subsection{Decision problem}
We study a decision problem in which the consequence $c \in \C$ of various alternative actions $a\in\mathcal{A}$ depend on the parameterization $\theta\in \Theta$ of an economic model. A consequence function $\rho: \A \times \Theta \mapsto \C$  details the consequence of action $a$ under parameters $\theta$:
\begin{align*}
c = \rho(a, \theta).
\end{align*}

A decision-maker ranks consequences according to a utility function $u: \C \mapsto \R$, where higher values are more desirable. The structure of the decision problem $(\A, \Theta, \C, \rho, u)$ is known, but the true parameterization $\theta_0$ is uncertain. As a result, the consequences of a particular action are ambiguous. An observed sample of data $\psi \in \Psi$, however, provides a signal about the true parameters, as $P_{\theta}$ -- the sampling distribution of $\psi$ -- differs by $\theta$. A statistical decision function (SDF) $\delta: \Psi \mapsto \mathcal{A}$ is a procedure that determines an action for each possible realization of the sample.\\

In our application, we study the bus replacement problem with unknown future mileage utilizations. The decision problem is dynamic, so a decision-maker acts by committing to a plan that specifies whether to maintain or replace a bus in any possible future scenario. The consequences of executing a particular plan are a stream of maintenance costs, aggregated by its discounted sum of utilities. A plan's total utility is determined by the true distribution of the bus mileage utilization. The optimal decision rule based on an RMDP depends on the observed sample of past mileage transitions as the sample informs the construction of the ambiguity set. So, each RMDP is one example of an SDF for the bus replacement problem. We consider many RMDPs with varying levels of robustness and thus analyze a whole class of SDFs.\\

Statistical decision theory provides the framework to compare the performance of alternative decision functions  $\delta \in \Gamma$. The utility achieved by any $\delta$ is a random variable before realizing $\psi$. Thus, \citet{Wald.1950} suggests measuring the performance of $\delta$ at all possible parametrizations $\theta$ by computing the expected utility with respect to its induced sampling distribution $P_{\theta}$:
\begin{align*}
  \E_{\theta}\left[u\left(\rho(\delta(\psi), \theta)\right)\right] = \int_\Psi u\left(\rho(\delta(\psi), \theta)\right) d P_{\theta}(\psi).
\end{align*}
In general, no single decision function yields the highest expected utility for all possible parameterizations. In this case,  determining the best decision function  $\delta^*$ is not straightforward. Still, decision theory proposes various criteria \citep{Gilboa.2009,Marinacci.2015} to aggregate the performance of a decision function at all possible parameterizations. At the most fundamental level, any decision function is admissible if another function does not exist, whose expected utility is always at least as high. In most cases, several decision functions are admissible, and thus additional optimality criteria are needed. Our analysis explores three of the most common decision criteria:  (1) maximin, (2) minimax regret, and (3) subjective Bayes.\\

Following the maximin decision criterium \citep{Wald.1950}, we determine the optimal decision function by computing the minimum expected performance for each decision function over all points in the parameter space. We then choose the one with the highest minimum performance. Stated concisely,
\begin{align*}
\delta^*= \argmax_{\delta \in \Gamma } \min_{\theta\in \Theta} \E_{\theta}\left[u(\rho\left(\delta(\psi), \theta\right))\right].
\end{align*}

For the minimax regret criterion \citep{Niehans.1948}, we compute the maximum regret for each decision function over all points in the parameter space. The regret of choosing a decision function for any realization of $\theta$ is the difference between the maximum possible performance, where the true parameterization informs the decision, and its actual performance. We then select the decision function with the lowest maximum regret. Thus, the minimax regret criterion solves:
\begin{align*}
\delta^* =  \argmin_{\delta \in \Gamma } \max_{\theta\in \Theta}  \underbrace{\left[\max_{a \in \A}  u(\rho\left(a, \theta\right))  - \E_{\theta}\left[u(\rho\left(\delta(\psi), \theta\right)) \right]\right]}_{\text{regret}}.
\end{align*}

Subjective Bayes \citep{Savage.1954} requires a subjective probability distribution $f_{\theta}$ over the parameter space. Then, we select the decision function with the highest expected subjective utility:
\begin{align*}
\delta^* = \argmax_{\delta \in \Gamma }  \int_{\theta} \E_{\theta}\left[u(\rho\left(\delta(\psi), \theta\right))\right]df_{\theta}.
\end{align*}

\subsection{Urn example}
We now illustrate the key ideas that allow us to compare as-if and robust decision-making using statistical decision theory with an urn example. As in our empirical application, we study a whole class of statistical decision functions. We first compare the performance of two distinct alternatives and then determine the optimal function within the class.\\

We consider an urn with black $b$ and white $w$ balls where the true share of black balls $\theta_0$ is unknown. In this example, the action constitutes a guess $\tilde{\theta}$ about $\theta_0$ after drawing a fixed number of $n$ balls at random with replacement. The parameter and action space both correspond to the unit interval $\Theta = \A = [0, 1]$. \\

If the guess matches the true share, we receive a payment of one. On the other hand, the payment is reduced by the squared error in case of a discrepancy. Thus, the consequence function takes the following form:
\begin{align*}
\rho(\tilde{\theta}, \theta_0) = 1 - (\tilde{\theta} - \theta_0)^2.
\end{align*}
Going forward, we assume a linear utility function and directly refer to the monetary consequences of a guess as its utility. The sample space is $\Psi = \{b, w\}^n$ where a sequence $(b, w, b, \hdots, b)$ of length $n$ is a typical realization of $\psi$. The observed number of black balls $r$ among the $n$ draws in a given sample $\psi$ provides a signal about $\theta_0$. The sampling distribution for the possible number of black balls $R$ takes the form of a probability mass function (PMF):
\begin{align*}
\Pr(R = r) = \left(\begin{array}{c} n \\ r \end{array} \right)\, (\theta_0)^r\, (1 - \theta_0)^{n-r}.
\end{align*}
Any function $\delta:  \{b, w\}^n \mapsto [0, 1]$ that maps the number of black draws to the unit interval is a possible statistical decision function.\\

We focus on the following class of decision functions $\delta \in \Gamma$, where each $\lambda$ indexes a particular decision function:
\begin{align*}
 \delta_\lambda(r) = \lambda\,\left(\frac{r}{n}\right)  + (1 - \lambda)\,\left(\frac{1}{2}\right),\qquad\text{for some}\quad 0 \leq \lambda \leq 1.
\end{align*}
The empirical share of black balls in the sample $r / n$ provides the point estimate $\hat{\theta}$. The decision functions in $\Gamma$ specify the guess as a weighted average between the point estimate and the midpoint of the parameter space. The larger $\lambda$ is, the more weight is put on the point estimate. At the extremes, the guess is either the point estimate ($\lambda = 1$) itself or fixed at $0.5$ ($\lambda = 0$).\\

We begin by comparing the performance of the two decision functions with $\lambda = 1$ and $\lambda = 0.9$. We refer to the former as the as-if decision function (ADF), as it announces the point estimate as if it is the true parameter. For reasons that will later become clear, we identify $\lambda=0.9$ as the robust decision function (RDF). Following \citet{Wald.1950}, we evaluate their relative performance by aggregating the vector of expected payoffs over the unit interval using the different decision-theoretic criteria. We set the number of draws $n$ to $50$.\\

Figure \ref{Calculation of expected payoff} shows the sampling distribution of the number of black balls $R$ and the associated payoff of following the two decision functions for each possible draw. The true, but unknown, share in this example is $40\%$, i.e. $\theta_0 = 0.4$. The RDF outperforms the ADF for realizations of the point estimates smaller than its true value due to the shift towards $0.5$. At the same time, the ADF leads to a higher payoff at the center of the distribution.\\

\begin{figure}[h!]\centering
\scalebox{0.75}{\includegraphics{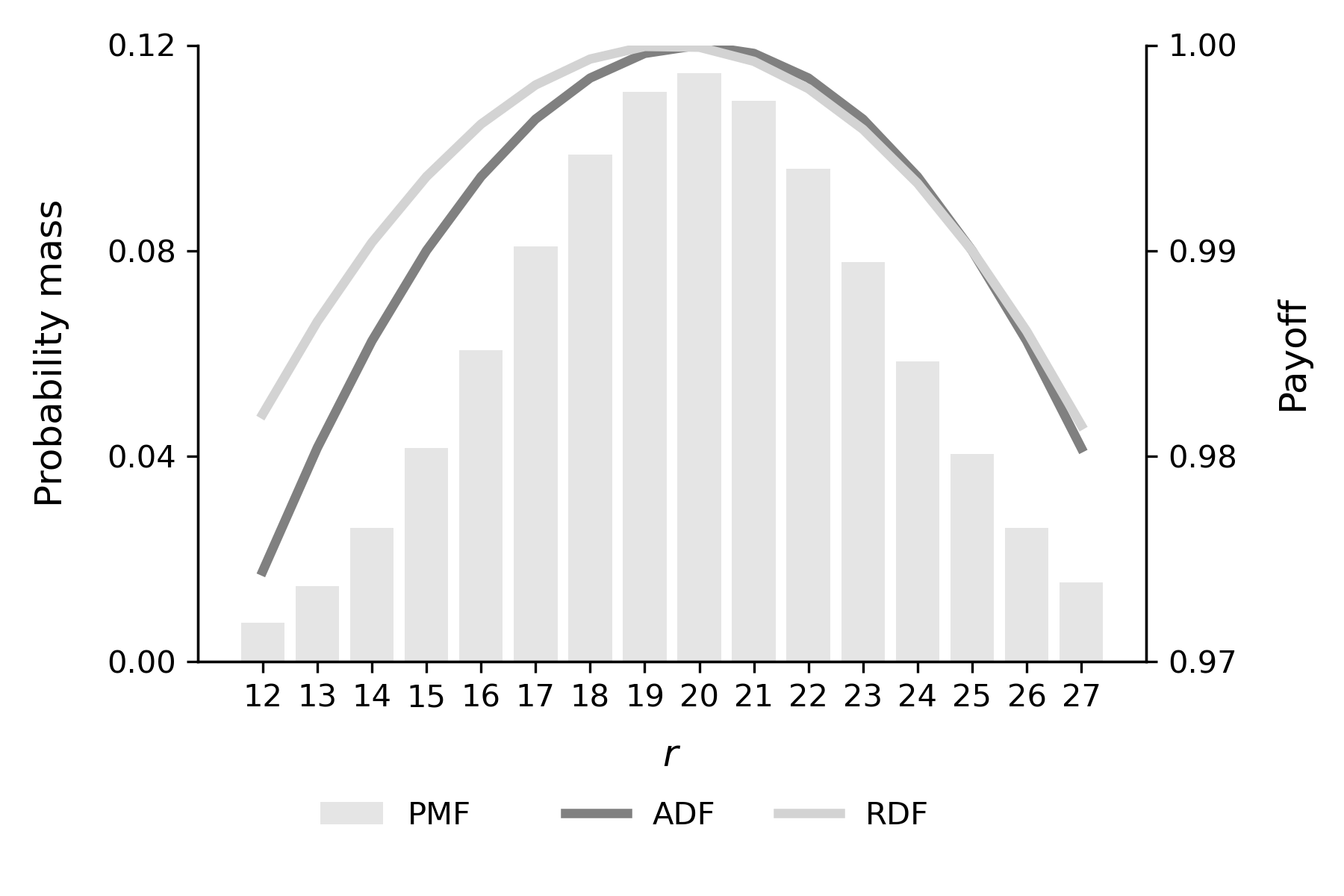}}
\caption{Calculating the expected payoff}\label{Calculation of expected payoff}
\end{figure}\FloatBarrier

Figure \ref{Measurement of performance} shows the expected payoff for varying shares $\theta$ of black balls in the urn. On the left, we show the expected payoff at two selected points. While the ADF performs better than the RDF at $\theta = 0.1$, the opposite is true at $\theta = 0.4$. Thus, both decision functions are admissible, as neither outperforms the other for all possible true shares. On the right, we trace the expected payoff of both functions over the whole parameter space. Although the RDF outperforms the ADF for shares in the center of the parameter space, it performs worse at the boundaries. Overall, the performance of the RDF is more balanced across the whole parameter space, which motivates its name.

\begin{figure}[h!]\centering
\subfloat[Expected payoff]{\scalebox{0.50}{\includegraphics{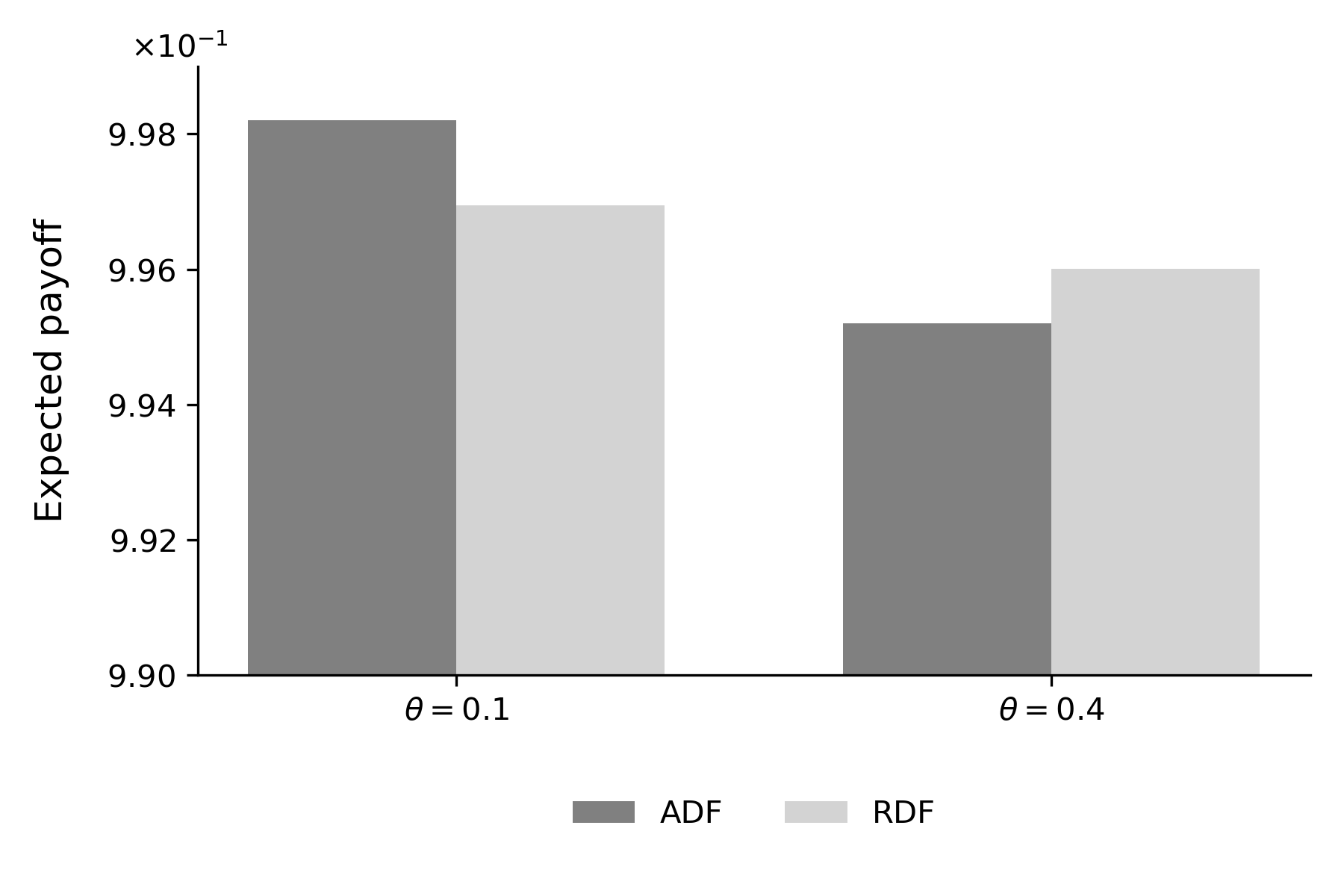}}}\hspace{0.3cm}
\subfloat[Performance]{\scalebox{0.50}{\includegraphics{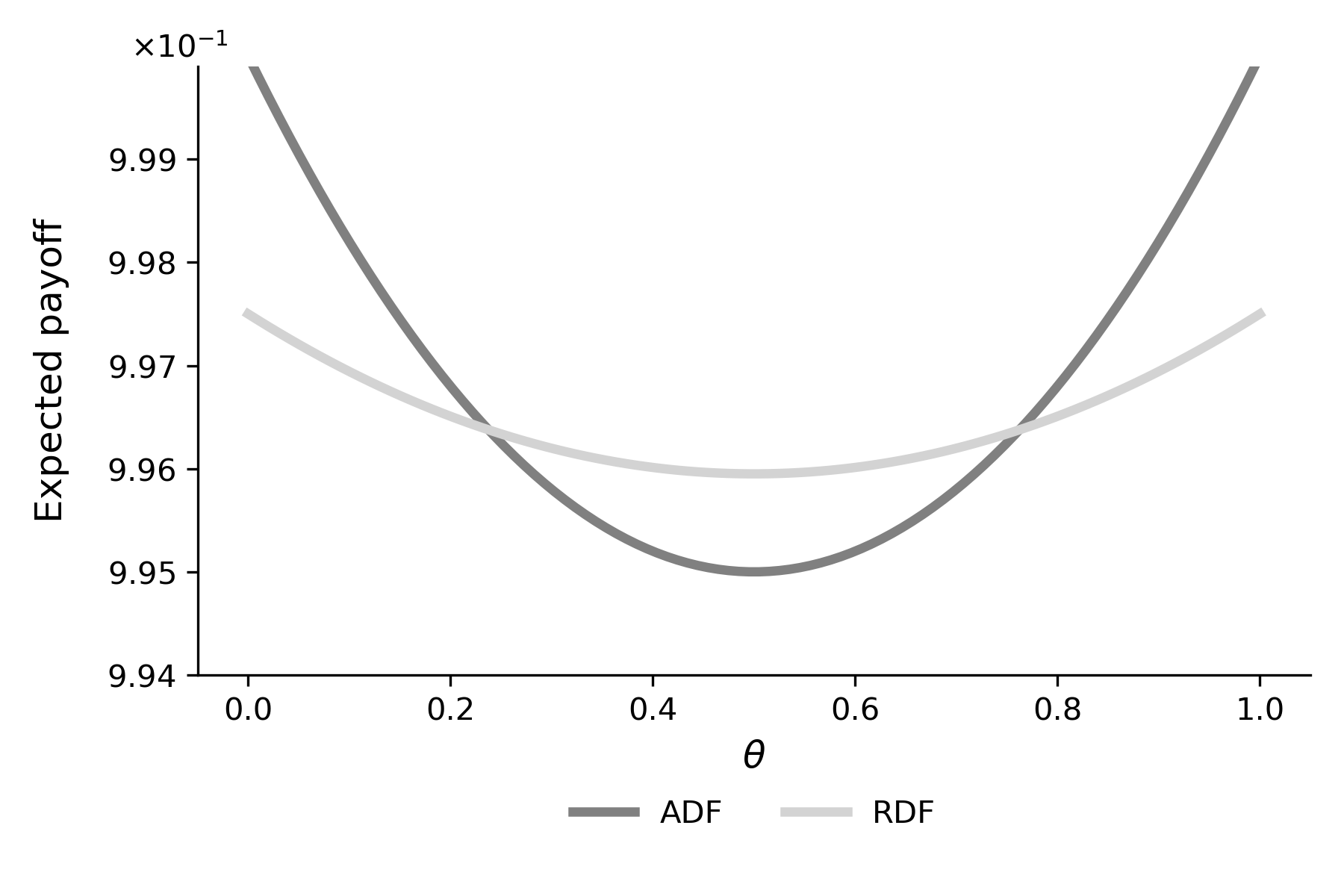}}}
\caption{Evaluating decision functions}\label{Measurement of performance}
\end{figure}\FloatBarrier

Figure \ref{Ranking of decision functions} ranks the two functions according to different decision-theoretic criteria. Both decision functions have their lowest expected payoff at $\theta = 0.5$. As the RDF outperforms its ADF alternative at that point, the RDF is preferred based on the maximin and minimax regret criteria. The maximin and minimax regret criteria are identical in this setting, as the payoff at the true share is constant across the parameter space. Using the subjective Bayes criterion with a uniform prior, we select the ADF, as its better performance at the boundaries of the parameter space is enough to offset its worse performance in the center.

\begin{figure}[h!]\centering
\scalebox{0.75}{\includegraphics{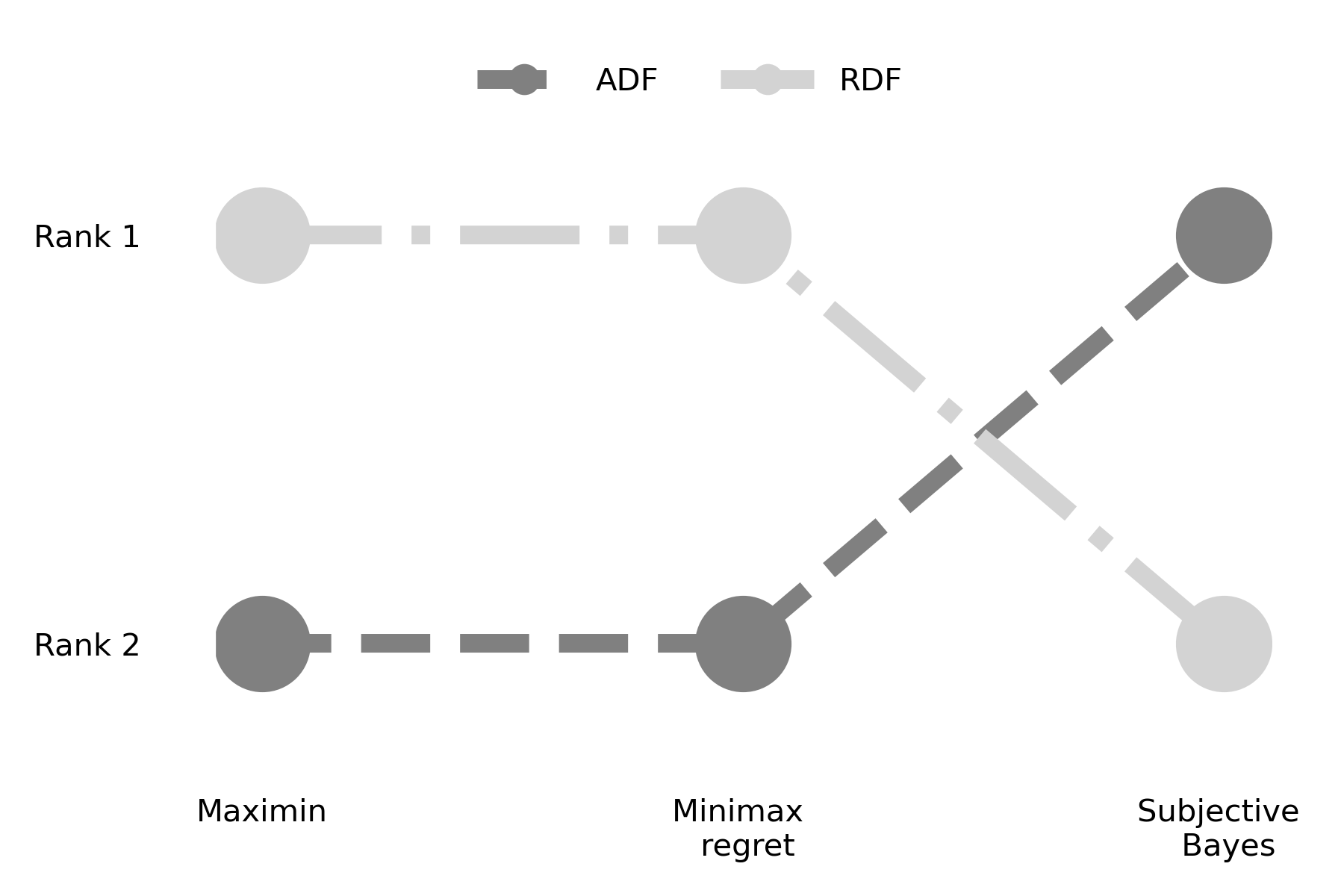}}
\caption{Ranking of decision functions}\label{Ranking of decision functions}
\end{figure}\FloatBarrier

Returning to the whole set of decision functions, we can construct the optimal statistical decision function in $\Gamma$ for the alternative criteria by varying $\lambda$ to maximize the relevant performance measure. For example, Figure \ref{Optimality of decision functions for urn} shows the minimum and the uniformly weighted performance for varying $\lambda$.

\begin{figure}[h!]\centering
\scalebox{0.75}{\includegraphics{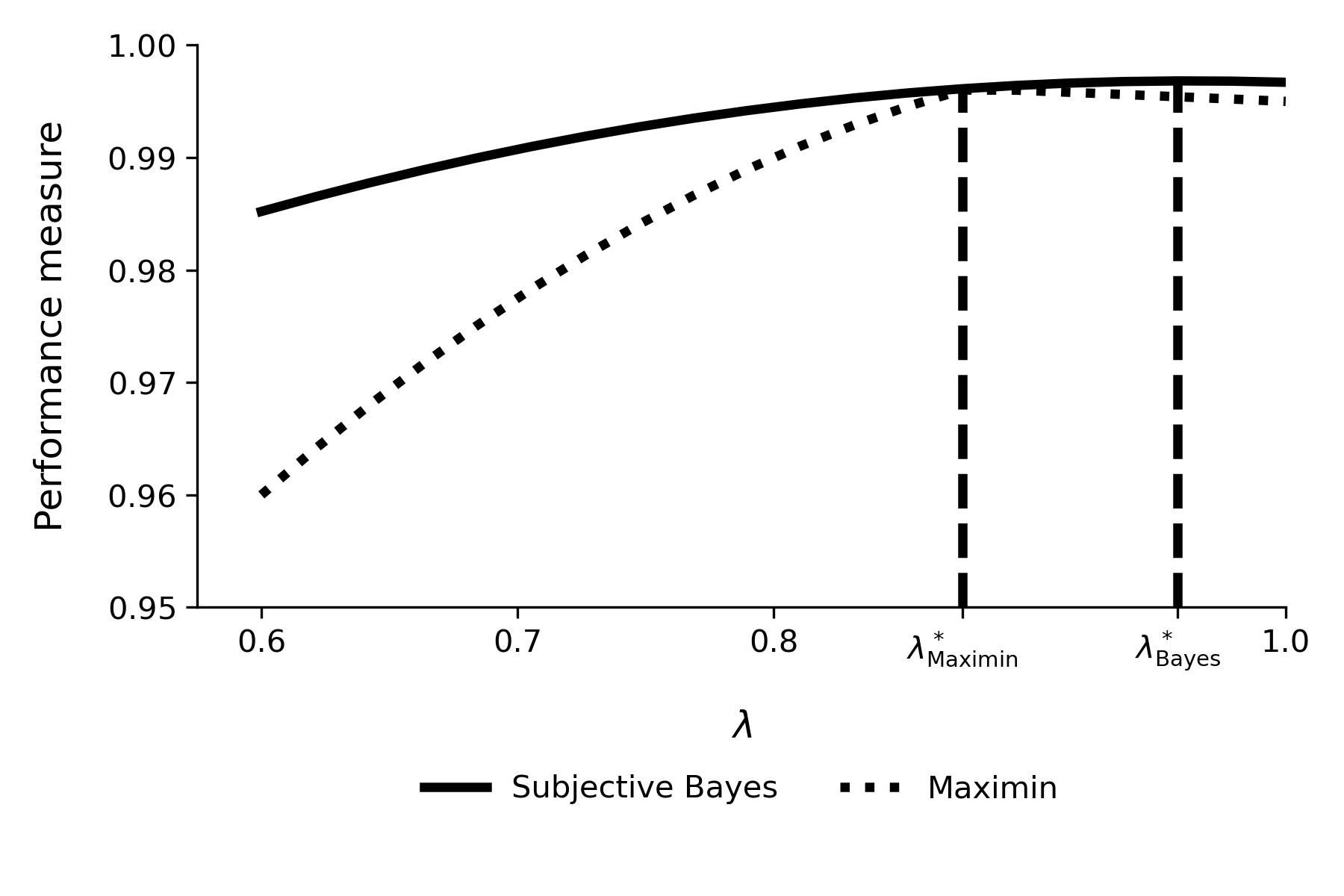}}\vspace{-0.9cm}
\begin{center}
\begin{minipage}[t]{0.5\columnwidth}
\item \scriptsize{\textbf{Notes:} We omit the performance measure for the minimax regret criterion, as $\lambda^*_{\text{Maximin}} = \lambda^*_{\text{Regret}}$ in this setting as noted earlier.}
\end{minipage}
\end{center}
\caption{Optimality of decision functions}\label{Optimality of decision functions for urn}
\end{figure}\FloatBarrier

Neither of our two decision functions analyzed earlier turns out to be optimal, as $\lambda^*_{\text{Bayes}} \approx 0.96$ and $\lambda^*_{\text{Maximin}} \approx 0.87$. Overall, the performance measure is more sensitive to the choice of $\lambda$ under the maximin criterion than under subjective Bayes.\\

In summary, the urn example illustrates the performance comparison of alternative decision functions over the whole parameter space. It shows how to construct an optimal decision function within a class for alternative decision-theoretic criteria. Next, we move to the more involved setting of a sequential dynamic decision problem with ambiguous transitions that we analyze in our application.

\FloatBarrier\section{Data-driven robust Markov decision problem}\label{Conceptual framework}
We now outline the framework of an RMDP for the analysis of sequential decision-making in light of model ambiguity. From the perspective of statistical decision theory, any RMDP with a fixed level of robustness is a statistical decision function. Once a sample of transitions is available, we construct the ambiguity set of a given size and solve the RMDP for a robust decision rule.\\

We first present the general setup of an RMDP and discuss the construction of the ambiguity set. We then turn to the solution approach and describe our decision-theoretic analysis to determine the optimal level of robustness. Throughout, we address the new challenges of analyzing an RMPD as opposed to a standard MDP. In line with our application, we discuss an infinite horizon model in discrete time, stationary utility and transition probabilities, and discrete states and actions.\footnote{See \citet{Puterman.1994} for a textbook introduction to the standard MDP and \citet{Rust.1994} for a review of MDPs in economics and structural estimation.}\\

We focus our exposition on ambiguity in the transition dynamics of the Markov decision process. We do not address uncertainty about the parameters of the reward functions. Although our central insight to use statistical decision theory to determine the optimal level of robustness is also relevant for the parameters of the reward functions, we do not address uncertainty pertaining to these parameters, as each setting introduces its unique computational challenges \citep{Mannor.2019}.

\FloatBarrier\subsection{Setting}
We consider the following decision problem. At time $t = 0, 1, 2, \hdots$ a decision-maker observes the state of their environment $s_t \in \ST$ and chooses an action $a_t$ from the set of admissible actions $\A$. The decision has two consequences. It creates an immediate utility $u(s_t, a_t)$, and the environment evolves to a new state $s_{t+1}$. The transition from $s_t$ to $s_{t+1}$ is affected by the action, and governed by a transition probability distribution $p(s_t, a_t)$.\\

Decision-makers take the future consequences of the current action into account. While a decision rule $d_t$ specifies the planned action for all possible states within period $t$, a policy $\pi =\{ d_0, d_1, d_2, \hdots \}$ is a collection of decision rules and specifies all planned actions for all time periods.\\

Figure \ref{Timing} depicts the timing of events in the decision problem. At the beginning of period $t$, a decision-maker learns about the utility of each alternative, chooses one according to the decision rule $d_t$, and receives its immediate utility. Then, the state evolves from $s_t$ to $s_{t+1}$, and the process repeats itself in $t + 1$.\\
\begin{figure}[h!]\centering


\begin{tikzpicture}[node distance=2cm]
\tikzstyle{startstop} = [circle, rounded corners, minimum width=0.6cm, minimum height=0.3cm,text centered, draw=black]
[
->,
>=stealth',
auto,node distance=3cm,
thick,
main node/.style={circle, draw, font=\sffamily\Large\bfseries}
]
\tikzstyle{arrow} = [thick,->,>=stealth]]
\tikzstyle{darrow} = [dotted,->,>=stealth]]

\node (r0) [startstop, xshift = -3cm, draw = none] {};
\node (r999) [startstop, xshift = 13cm, draw = none] {};

\node (r1) [startstop, xshift = -1cm] {\footnotesize $~\,s_t\,~$};
\node (r2) [startstop, xshift = 5cm] {\footnotesize $s_{t+1}$};  
\node (r3) [startstop, xshift = 11cm] {\footnotesize $s_{t+2}$}; 

\draw [arrow, dashed] (r0) -- node[anchor=south] {} (r1) ;
\draw [arrow] (r1) -- node[anchor=south] {} (r2) ;
\draw [arrow] (r2) -- node[anchor=south] {} (r3) ;
\draw [arrow, dashed] (r3) -- node[anchor=south] {} (r999) ;
t
\node (r4) [startstop, xshift = 0 cm, yshift = -2.5cm, inner sep = 0.08cm] {\footnotesize $~\,a_t\,~$ };
\node (r5) [startstop, xshift = 3.5
 cm, yshift = -2.5cm, inner sep = 0.08cm] {\footnotesize $~\,u_t\,~$ };
\node (r6) [startstop, xshift = 6 cm, yshift = -2.5cm, inner sep = 0.08cm] {\footnotesize $a_{t+1}$ };
\node (r7) [startstop, xshift = 9.5 cm, yshift = -2.5cm, inner sep = 0.08cm] {\footnotesize $u_{t+1}$ };

\draw [arrow] (r1) -- node[anchor=south] {} (r4) ;
\draw [arrow] (r2) -- node[anchor=south] {} (r6) ;

\draw[ thick](0,-2.05).. controls (0.75, -0.2) and (0.8,0)..(1.5, 0);
\draw [arrow] (r4) -- node[anchor=south] {} (r5) ;
\draw[ thick](6,-2.05).. controls (6.75, -0.2) and (6.85,0)..(7.5, 0);
\draw [arrow] (r6) -- node[anchor=south] {} (r7) ;

\node(r8)[startstop, xshift = - 1.3cm, yshift = -1.3cm, draw =none, align=center] {\footnotesize decide \\ \footnotesize $d_{t}$ };
\node(r9)[startstop, xshift= 4.6cm, yshift = -1.3cm, draw =none, align = center] {\footnotesize decide \\ \footnotesize $d_{t + 1}$ };
\node(r10)[startstop, yshift = 0.5cm, xshift = 1.7cm, draw =none, align=center ] {\footnotesize transition \\ \footnotesize $p(s_t, a_t)$};
\node(r11)[startstop, yshift = 0.5cm, xshift = 8cm, draw =none, align=center ] {\footnotesize transition \\ \footnotesize $p(s_{t+1}, a_{t+1})$};
\node(r12)[startstop, yshift = -2cm, xshift = 1.7cm, draw =none, align=center ] {\footnotesize receive\\ \footnotesize $u(s_t, a_t)$};
\node(r13)[startstop, yshift = -2cm, xshift = 7.7cm, draw =none, align=center ] {\footnotesize receive\\ \footnotesize $u(s_{t+1}, a_{t+1})$};

\end{tikzpicture}
\caption{Timing of events}\label{Timing}
\end{figure}
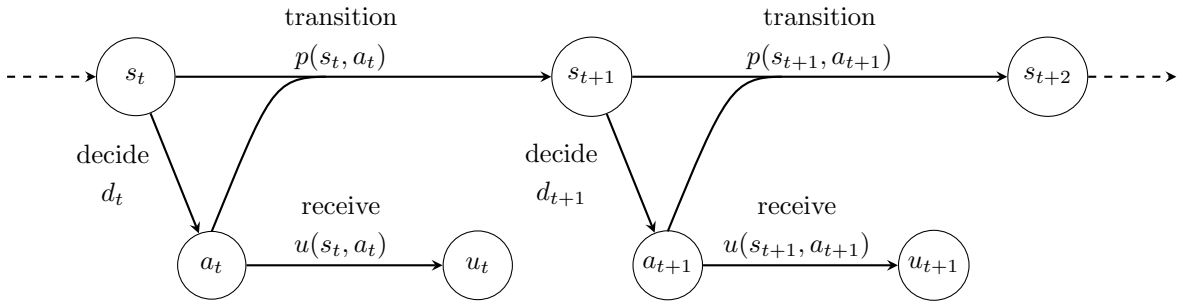\FloatBarrier

In a standard Markov decision process (MDP), a single transition probability distribution $p(s_t, a_t)$ is associated with each state-action pair. This distribution is assumed to be known, and thus the MDP incorporates risk only. In an RMDP, there is a whole set of distributions associated with each state-action pair collected in an ambiguity set $p(s_t, a_t) \in \mathcal{P}(s_t, a_t)$. For a particular RMDP, the ambiguity set is assumed to be known, and thus the RMDP incorporates risk for a given distribution and ambiguity about the true distribution.\\

In a standard MDP, the objective of a decision-maker in state $s_t$ at time $t$ is to choose the optimal policy $\pi^*$ from the set of all possible policies $\Pi$ that maximizes their expected total discounted utility $\tilde{v_t}^{\pi^*}(s_t)$ as formalized in Equation (\ref{Objective risk}):
\begin{align}\label{Objective risk}
\tilde{v_t}^{\pi^*}(s_t) \equiv \max_{\pi\in\Pi} \E^\pi\left[\sum^{\infty}_{r = 0}  \delta^{t+r} u(s_{t+r}, d_{t+r}(s_{t+r}))\right].
\end{align}
The exponential discount factor $\delta$ parameterizes a taste for immediate over future utilities. The superscript of the expectation emphasizes that each policy induces a different probability distribution over sequences of possible futures. As long as transition probabilities used to construct the policy are in fact correct, the standard value function $\tilde{v_t}^{\pi^*}(s_t)$ measures the performance of the optimal policy.\\

In an RMDP, the goal is to implement an optimal policy that maximizes the expected total discounted utility under a worst-case scenario. Given the ambiguity about the transition dynamics, a policy induces a whole set of probabilities over sequences of possible future utilities $\mathcal{F}^\pi$, and the worst-case realization determines its ranking. The formal representation of the decision-maker's objective is Equation (\ref{Objective risk and ambiguity}):
\begin{align}\label{Objective risk and ambiguity}
v_t^{\pi^*}(s_t) \equiv \max_{\pi\in\Pi} \left\{\min_{\vec{P} \in \mathcal{F}^\pi}\E^\vec{P}\left[\sum^{\infty}_{r = 0}  \delta^{t+r} u(s_{t+r}, d_{t+r}(s_{t+r}))\right]\right\}.
\end{align}

We consider a setting where historical data provides information about the transition dynamics. In the data-driven standard  MDP, the empirical probabilities $\hat{p}(s_t, a_t)$ serve as a plug-in for the truth, and the solution of the MDP provides an as-if decision rule. In a data-driven RMDP, the empirical probabilities are used to construct the ambiguity sets for the transitions, and the solution of the RMDP provides a robust decision rule.\\

We follow \cite{Ben-Tal.2013} and create the ambiguity sets using statistical hypothesis testing. We restrict attention to distributions we cannot reject with a certain level of confidence $\omega \in [0, 1]$ around the empirical probabilities and collect them in an estimated ambiguity set $\hat{\mathcal{P}}(s_t, a_t; \omega)$. Different values of $\omega$ result in different RMDPs, each with its own statistical decision function. Two special cases stand out. First, if $\omega = 0$, then a decision-maker treats the empirical probabilities as if they are correct. This case captures the notion of as-if decision-making. Second, for $\omega = 1$, a robust decision-maker considers the worst-case scenario over the whole probability simplex at each state-action pair when constructing the optimal policy.

\FloatBarrier\subsection{Solution}\label{Solution}
In a standard MDP, the objective is to maximize the expected total discounted utility as formalized in Equation (\ref{Objective risk}). This requires evaluating the performance of all policies based on all possible sequences of utilities and the probability that each occurs. Fortunately, the stationary Markovian structure of the problem implies that the future looks the same whether the decision-maker is in state $s$ at time $t$ or any other point in time. The only variable that determines the value to the decision-maker is the current state $s$. Thus, the optimal policy is stationary as well \citep{Blackwell.1965}, and the same decision rule is used in every period. The value function is independent of time and of the solution to the following Bellman equation:
\begin{align}\label{Bellman equation}
\tilde{v}(s) = \max_{a \in \A} \biggl[u(s, a) + \delta\thin \int \thin \tilde{v}(s^\prime)\thin \hat{p}(ds^\prime| s , a)\biggr].
\end{align}
The as-if decision rule is recovered from Equation (\ref{Bellman equation}) by finding the value $a \in \A$ that attains a maximum for each $s \in \ST$.\\

Let $\mathbb{V}$ denote the set of all bounded real value functions on $\ST$. Then, the Bellman operator $\tilde{\Lambda} : \mathbb{V} \rightarrow \mathbb{V}$ is defined as follows: For all $w\in\mathbb{V}$
\begin{align}
\label{classic_bellman}
  \tilde{\Lambda}(w)(s) = \max_{a \in \A} \biggl[u(s, a) + \delta\thin \int \thin w(s^\prime)\thin \hat{p}(ds^\prime| s , a)  \biggr], \quad s\in \ST.
\end{align}
Under mild conditions, $\tilde{\Lambda}$ is a contraction mapping and allows to compute the value function $\tilde{v}(\cdot)$ as its unique fixed point \citep{Denardo.1967}.\\

For an RMDP, where transition probabilities are ambiguous, the contraction mapping property of the Bellman operator and the optimality of a stationary deterministic Markovian decision rule both require the assumption of rectangularity of $\mathcal{F}^\pi$ \citep{Iyengar.2005,Nilim.2005}. As the realization of any particular distribution in a state-action pair does not affect future realizations, rectangularity is a form of an independence assumption. The uncertainty is uncoupled across states and actions. This approach rules out any kind of learning about future ambiguity from past experiences due to, for example, a common source of uncertainty across states. While restrictive, most applications rely on the rectangularity assumption, as general notions of coupled uncertainties are intractable \citep{Wiesemann.2013}.\footnote{See \cite{Mannor.2016} and \cite{Goyal.2020} for recent attempts to introduce milder rectangularity conditions.}\\

We now develop the formal definition of rectangularity. Let $\mathcal{M}(\ST)$ denote the set of all probability distributions on $\ST$. Then, the set of all conditional transition probability distributions associated with any decision rule $d$ is given by:
\begin{align*}
\mathcal{F}^{d} & = \{p: \ST \rightarrow \mathcal{M}(\ST) \thin | \thin \forall s \in \ST,\thin p(s) \in \hat{\mathcal{P}}(s, d(s); \omega)\}.
\end{align*}
For every state $s \in \ST$, the next state can be determined by any $p \in \hat{\mathcal{P}}(s, d(s); \omega)$.\\

A policy $\pi$ now induces a set of probability distributions $\mathcal{F}^\pi$ on the set of all possible histories $\mathcal{H}$. Any particular history $h = (s_0, a_0, s_1, a_1, \hdots)$ can be the result of many possible combinations of transition probabilities. Rectangularity imposes a structure on the combination possibilities.

\begin{Assumption}\label{Rectangularity}\textbf{Rectangularity} The set $\mathcal{F}^\pi$ of probability  distributions associated with a policy $\pi$ is given by
\begin{align*}
\mathcal{F}^\pi & = \bigg\{\vec{P} \mid \forall\, h\in \mathcal{H}:\, \vec{P}(h) =\prod^{\infty}_{t = 0}  p(s_{t+1}|s_t, a_t), \textup{ with } p(s_t, a_t) \in \hat{\mathcal{P}}(s_t, d_t(s_t); \omega) \textup{ for } t = 0, 1, \hdots \bigg\} \\
&= \mathcal{F}^{d_0} \times \mathcal{F}^{d_1} \times \mathcal{F}^{d_2} \times \hdots = \prod^{\infty}_{t = 0} {F}^{d_t},
\end{align*}
where the notation simply denotes that each element in $\mathcal{F}^\pi$ is a product of $p \in\mathcal{F}^{d_t}$, and vice versa \citep{Iyengar.2005}.
\end{Assumption}

Assumption \ref{Rectangularity} formalizes the idea that ambiguity about the transition probability distribution is uncoupled across states and time. All elements of the ambiguity sets can be freely combined to generate a particular history.\\

The objective when facing ambiguity is to implement a policy $\pi^*$ that maximizes the expected total discounted utility under a worst-case scenario as presented in Equation (\ref{Objective risk and ambiguity}). Under the rectangularity assumption, the decision-maker faces the same uncertainty, whether he is in state $s$ at time $t$ or any other point in time. Thus, the value function is independent of time and solely depends on the current state $s$. It is the solution to the robust Bellman equation (\ref{Robust Bellman equation}), where the future value is evaluated using the worst-case element in the ambiguity set \citep{Iyengar.2005}:
\begin{align}\label{Robust Bellman equation}
v(s) = \max_{a \in \A} \biggl[u(s, a) +  \delta\thin \underset{ p\in\hat{\mathcal{P}}(s, a; \omega)}{\min}\thin \int v(s^\prime)\thin p(ds^\prime|s, a)\biggr].
\end{align}
The robust decision rule is recovered from Equation (\ref{Robust Bellman equation}) by finding the value $a \in \A$ that attains a maximum for each $s \in \ST$ under the worst-case scenario for all distributions in the ambiguity set.\\

The robust Bellman operator on $\mathbb{V}$ follows directly: For all $w\in\mathbb{V}$
\begin{align}\label{robust Bellman operator}
  \Lambda(w)(s) = \max_{a \in \A} \biggl[u(s, a) + \delta\thin \min_{ p\in \hat{\mathcal{P}}(s, a; \omega)} \int w(s^\prime)\thin p(ds^\prime|s, a)  \biggr] \quad s\in\ST.
\end{align}
Algorithm \ref{Robust Value Iteration Algorithm} allows solving the RMDP by a robust version of the value iteration algorithm where $\kappa$ denotes a convergence threshold. The calculation of future values under the worst-case scenario is the key difference to the standard approach.

\vspace{0.5cm}\begin{algorithm}
\caption{\strut Robust Value Iteration Algorithm}\label{Robust Value Iteration Algorithm}
\begin{algorithmic}\vspace{0.3cm}
\State $ \textbf{Input:}\quad v \in \mathbb{V}, \kappa > 0$\\\\
For each $s \in \ST$, set $\hat{v}(s) = \underset{a \in \A}{\max}\left\{u(s, a) + \underset{ p\in \hat{\mathcal{P}}(s, a; \omega)}{\min}\int v(s^\prime)\thin p(ds^\prime|s, a)]\right\}$.\\\\
\While{$\left\| \, v  - \hat{v}\, \right\|_\infty > \kappa$}\\\\
$\qquad v = \hat{v}$ \\
$\qquad\forall\, s \in \ST$, set $\hat{v}(s) = \underset{a \in \A}{\max}\left\{u(s, a) + \underset{ p\in \hat{\mathcal{P}}(s, a; \omega)}{\min}\int v(s^\prime)\thin p(ds^\prime|s, a)]\right\}$\\
\EndWhile
\vspace{0.3cm}\end{algorithmic}
\end{algorithm}

\FloatBarrier\subsection{Evaluation}
The solution of an RMDP is tailored to the simultaneous worst-case realization of all distributions in all ambiguity sets. Although this conservative approach ensures a minimum performance over all distributions in the set, the performance of the robust decision rule in all other cases is disregarded. This indifference introduces a trade-off when determining the size of the ambiguity set \citep{Delage.2010}. The larger the set, the more scenarios for which a minimum performance is ensured. However, the robust rule's general performance suffers. This trade-off is particularly pronounced when the actual structure of the decision problem exhibits coupled uncertainties that are ignored in the construction of the robust rule to ensure its computational tractability.\\

Statistical decision theory allows us to navigate the trade-off and determine the optimal level of robustness. Each RMDP is a different statistical decision function, and we consider the whole class of statistical decision functions each indexed by $\omega\in [0 , 1]$. Adapting our urn example from earlier to accommodate setting of a data-driven RMDP, the parameter space corresponds to the set of transition probability distributions $\mathcal{L(\mathcal{S}, \mathcal{A})} = \{p:\mathcal{S} \times \mathcal{A} \rightarrow \mathcal{M}(\mathcal{S})\}$. We observe data on the transition probabilities and measure the actual performance of a robust decision rule $\eta(\hat{p}; p_0, \omega)$ as the discounted sum of utilities, which depends on the estimate of the transition probabilities $\hat{p}$, the true underlying probabilities $p_0$, and the confidence level $\omega$ used to construct the robust decision function. The standard decision-theoretic criteria translate to this setting as follows:
\begin{align*}\renewcommand{\arraystretch}{1.3}
\begin{array}{ll}
\text{\textbf{Maximin}}\qquad\qquad & \omega^* = \argmax_{\omega \in [0, 1]} \min_{p\in \mathcal{L(\mathcal{S}, \mathcal{A})}} \E_{p}\left[\eta(\hat{p} ; p, \omega)\right]\\
\text{\textbf{Minimax regret}}\qquad\qquad & \omega^* = \argmin _{\omega \in [0, 1]} \max_{p\in \mathcal{L(\mathcal{S}, \mathcal{A})}}\left[\max_{\tilde{\omega}\in[0, 1]}   \eta(p ; p, \tilde{\omega}) -\E_{p}\left[\eta(\hat{p} ; p, \omega)\right] \right] \\
\text{\textbf{Subjective Bayes}}\qquad\qquad &w^* =\argmax_{\omega \in [0, 1]} \int_{\mathcal{L(\mathcal{S}, \mathcal{A})}} \E_{p}\left[\eta(\hat{p} ; p, \omega)\right] \, d f_p \\
\end{array}
\end{align*}
Note that even for genuinely uncoupled uncertainties, the maximin criterion does not automatically select the most robust statistical decision function ($\omega = 1$). This particular decision function is based on the worst-case scenario over the full probability simplex at each state-action pair. In fact, the worst-case decision function might not be admissible in particular settings where it is weakly dominated by the as-if (or some other) decision function. Suppose, for example, the true distribution corresponds to the worst-case distributions. In this case, the distribution of sampled transitions is degenerate, as the worst-case scenario at each state-action pair is the certain transition to the state with the lowest future value \citep{Nilim.2005}. Thus, the as-if and worst-case decision functions share the same performance. For all other true distributions, the as-if decision function may very well outperform the worst-case decision function if the sampled data is sufficiently informative.

\FloatBarrier\section{Bus replacement problem}\label{Bus replacement problem}
We now study robust decision-making in the seminal bus replacement problem. First, we discuss the general setting and the details of the computational implementation. Second, we conduct an ex-post analysis of robust decision rules constructed for the observed sample of mileage transitions analyzed in \citet{Rust.1987}. Third, considering the situation before any data is realized, we conduct an ex-ante analysis of robust decision functions with varying levels of robustness over the whole probability simplex, which allows us to determine the optimal level of robustness using statistical decision theory.

\FloatBarrier\subsection{Setting}
The bus replacement model is set up as a regenerative optimal stopping problem \citep{Chow.1971}.
It is motivated by the sequential decision problem of a maintenance manager, Harold Zurcher, for a fleet of buses. He makes repeated decisions about their maintenance to maximize the expected total discounted utility under a worst-case scenario. Each month $t$, a bus arrives at the bus depot in state $s_t = (x_t, \epsilon_t)$ described by its mileage since the last engine replacement $x_t$ and other signs of wear and tear $\epsilon_t$. He faces the decision to either conduct a complete engine replacement $(a_t = 1)$ or perform basic maintenance work $(a_t = 0)$. The cost of maintenance $c(x_t)$ increases with the mileage state, while the cost of replacement $RC$ remains constant. In the case of an engine replacement, the mileage state is reset to zero. Note that we do not attempt to describe Harold Zurcher's decision-making process. Instead, we are interested in how a generic decision-maker should make decisions in this setting.\\

The immediate utility of each action is given by:
\begin{align*}
u(a_t, x_t) + \epsilon_t(a_t) \quad \text{with} \quad u(a_t, x_t) = \begin{cases}
-RC   & a_t = 1 \\
-c(x_t) & a_t = 0.
\end{cases}
\end{align*}

Decisions are made in light of uncertainty about next month's state variables captured by their conditional distribution $p(x_t, \epsilon_t, a_t)$.\\

Although in this framework, the utility and consequently the value function is finite in each state, they are not uniformly bounded. This property, however, is a crucial assumption for the results of \cite{Blackwell.1965} and \cite{Denardo.1967} on the contraction property of the Bellman operator and the stationarity of the optimal policy in the standard MDP setting. For the original as-if analysis, \citet{Rust.1988} circumvents this problem by imposing conditional independence between the observable and unobservable state variables, i.e. $p(x_{t+1}, \epsilon_{t+1}|  x_t, \epsilon_t, a_t) = p(x_{t+1}| x_t, a_t)\thin q(\epsilon_{t+1}|x_{t+1})$, and assuming that the unobservables $\epsilon_t(a_t)$ are independent and identically distributed according to an extreme value distribution with mean zero and scale parameter one. These two assumptions, together with the additive separability between the observed and unobserved state variables in the immediate utilities, ensure that the expectation of the next period's value function is independent of the time. The regenerative structure of the process implies that the transition probabilities in case of replacement in any mileage state correspond to the probabilities of maintenance in the zero mileage state. Therefore, the expected value function is the unique fixed point of a contraction mapping on the reduced space of mileage states only. In addition, the conditional choice probabilities $P(a_t | x_t)$ have a closed-form solution \citep{McFadden.1973}. We build on these results and extend them to our robust setting with ambiguous transition dynamics. The proof is available in Appendix \ref{The robust contraction mapping}.\\

In the analysis of the original bus replacement problem, the distribution of the monthly mileage transitions are estimated in a first step and used as plug-in components for the subsequent analysis. We extend the original setup and explicitly account for the ambiguity in the estimation. Following the arguments on the regenerative structure of the process above, we incorporate ambiguity in the RMDP with ambiguity sets conditional on the mileage states $x$ only. We construct ambiguity sets $\hat{\Pc}(x; \omega)$ based on the Kullback-Leibler divergence $D_{KL}$ \citep{Kullback.1951} that are statistically meaningful, computationally tractable, and anchored in empirical estimates $\hat{p}(x)$.\\

Our ambiguity set takes the following form for each mileage state $x$:
\begin{align*}
  \hat{\Pc}(x; \omega) = \biggl\{ q \in \mathring{\Delta}_{|J_x|}: D_{KL}(q \parallel \hat{p}(x)) = \sum_{i=1}^{|J_{x}|} q_i \ln\left(\frac{q_i}{\hat{p}(j_i|x)}\right) \leq \rho_{x}(\omega) \biggr\},
\end{align*}
where $J_x = \{j_1,\thin \dots,\thin j_{|J_x|}\}$ denotes the set of all states that have an estimated non-zero probability to be reached from $x$, $\mathring{\Delta}_{|J_{x}|} = \{p \in \mathbb{R}^{|J_{x}|}\,|\, p_i > 0 \text{ for all } i=1,\dots, |J_x| \text{ and } \sum_{i=1}^{|J_x|} p_i = 1\}$ is the interior of the $(|J_{x}| - 1)$ - dimensional probability simplex, and $\rho_{x}(\omega)$ captures the size of the set for the state $x$ with a given level of confidence $\omega$.\\

\cite{Iyengar.2002} and \cite{Ben-Tal.2013} provide the statistical foundation to calibrate $\rho_x(\omega)$ such that the true (but unknown) distribution $p_0$ is contained within the ambiguity set for a given level of confidence $\omega$. Let $\chi^2_{df}$ denote a chi-squared random variable with $df$ degrees of freedom, and let $F_{df}(\cdot)$ denote its cumulative distribution function with inverse $F^{-1}_{df}(\cdot)$. Then, the following approximate relationship holds as the number of observations $N_x$ for state $x$ tends to infinity \citep{Pardo.2005}:
\begin{align*}
  \omega & = \Pr [\,p_0\in \hat{\Pc}(x; \omega) \,] \\
         & \approx \Pr [\, \chi_{|J_x| - 1}^2 \leq 2 N_x\rho_x(\omega) \,]   \\
         & = F_{|J_x| - 1}(2 N_x \rho_x(\omega)).
\end{align*}
We can therefore calibrate the size of the ambiguity set based on the following relationship:
\begin{align}\label{Mapping}
  \rho_x(\omega) = \tfrac{1}{2N_x} F_{|J_x| - 1} ^ {-1} (\omega).
\end{align}

We use \citetalias{Rust.1987} original data to inform our computational experiments. His data consists of monthly odometer readings $x_t$ and engine replacement decisions $a_t$ for 162 buses. The fleet consists of eight groups that differ in their manufacturer and model. We focus on the fourth group of 37 buses with a total of 4,292 monthly observations. We discretize mileage into $78$ equally spaced bins of length $5,000$ and set the discount factor to $\delta=0.9999$.\\

Figure \ref{Distributions of observations} highlights the limited information about the true distribution of mileage utilization. It shows the number of observations available to estimate next month's utilization for different levels of accumulated mileage. While there are more than 1,150 observations on buses with less than 50,000 miles, there are only about 220 with more than 300,000.\\

\begin{figure}[h!]\centering
\scalebox{0.75}{\includegraphics{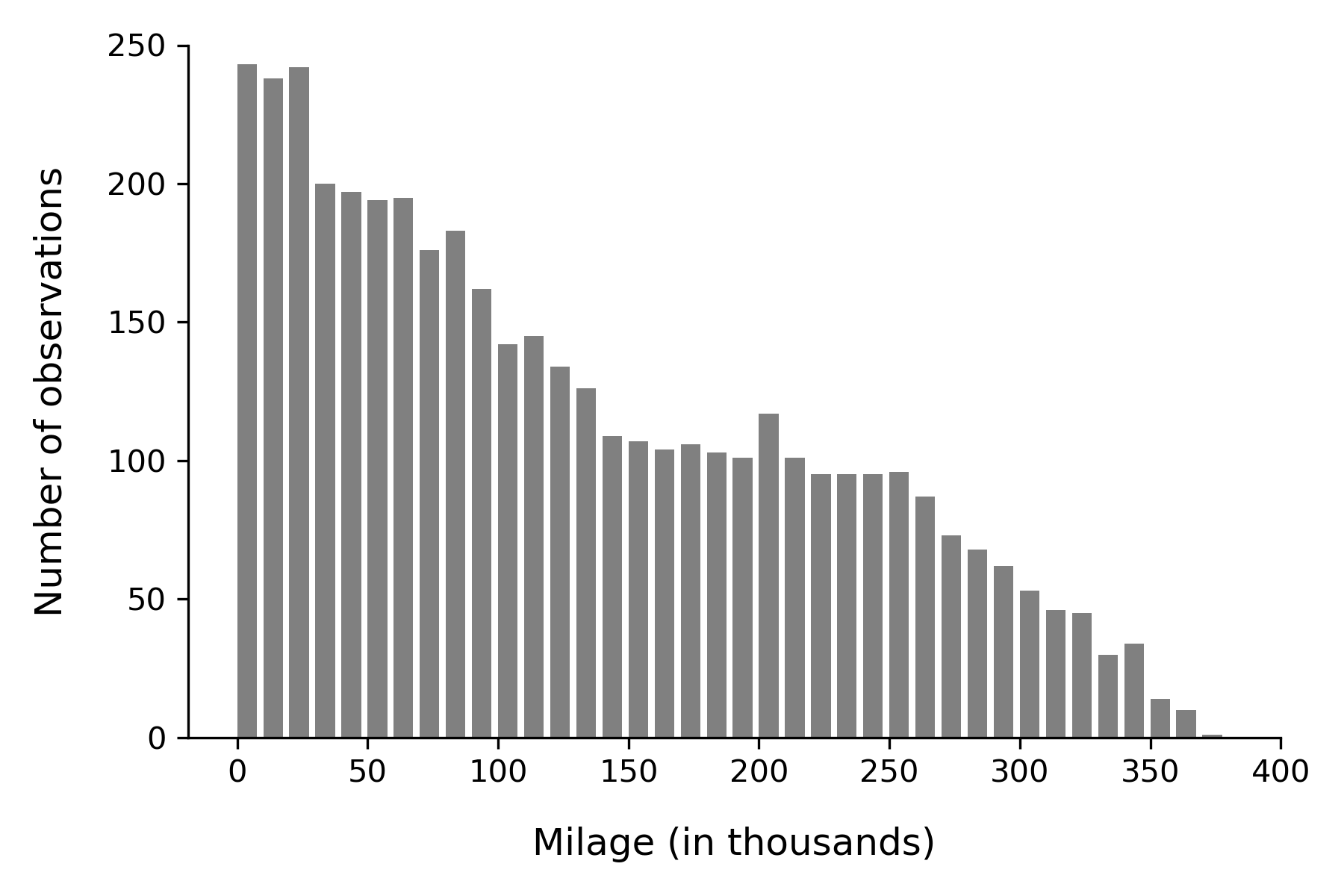}}
\caption{Distributions of observations}\label{Distributions of observations}
\end{figure}

We analyze a specific example of \citetalias{Rust.1987} bus replacement problem. We do not use his reported estimates of the maintenance and replacement costs. Given these estimates, decisions are mainly driven by the unobserved state variable $\epsilon_t$, and so ambiguity about the evolution of the observed state variable $x_t$ does not have a substantial effect on decisions. We ensure that a bus's accumulated mileage has a considerable impact on the timing of engine replacements by increasing the maintenance and replacement costs compared to their empirical estimates. Thus, we specify the following cost function $c(x_t) = 0.4\, x_t $ and set the replacement costs $RC$ to 50.\\

We solve the model using a modified version of the original nested fixed point algorithm (NFXP) \citep{Rust.1988}, and we determine the worst-case transition probabilities in each successive approximation of the fixed point. Given the size of the ambiguity set, we can determine the worst-case probabilities as the solution to a one-dimensional convex optimization problem \citep{Iyengar.2005,Nilim.2005}.\footnote{The core routines are implemented in our group's \citet{ruspy.2020} and \citet{robupy.2020} software packages and are publicly available.}

\subsection{Ex-post analysis}
We first study as-if and robust decision rules for \citetalias{Rust.1987} observed sample of mileage transitions. We present the estimated transition probabilities and the corresponding worst-case distributions. We then explore alternative decision rules based on several RMDPs, outline the resulting differences in maintenance decisions, and evaluate their relative performance under different scenarios.\\

Figure \ref{Estimated transition probabilities} shows the point estimates $\hat{p}$ for the transition probabilities of monthly mileage usage. We pool all 4,292 observations to estimate this distribution by maximum likelihood, and thus the probability of the next period's mileage utilization is the same for each state $x_t$. We only observe increases of at most $J = 3$ grid points per month. For about 60\% of the sample, monthly bus utilization is between 5,000 and 10,000 miles. Very high usage of more than 10,000 miles amounts to only 1.2\%.\\

\begin{figure}[h!]\centering
\scalebox{0.75}{\includegraphics{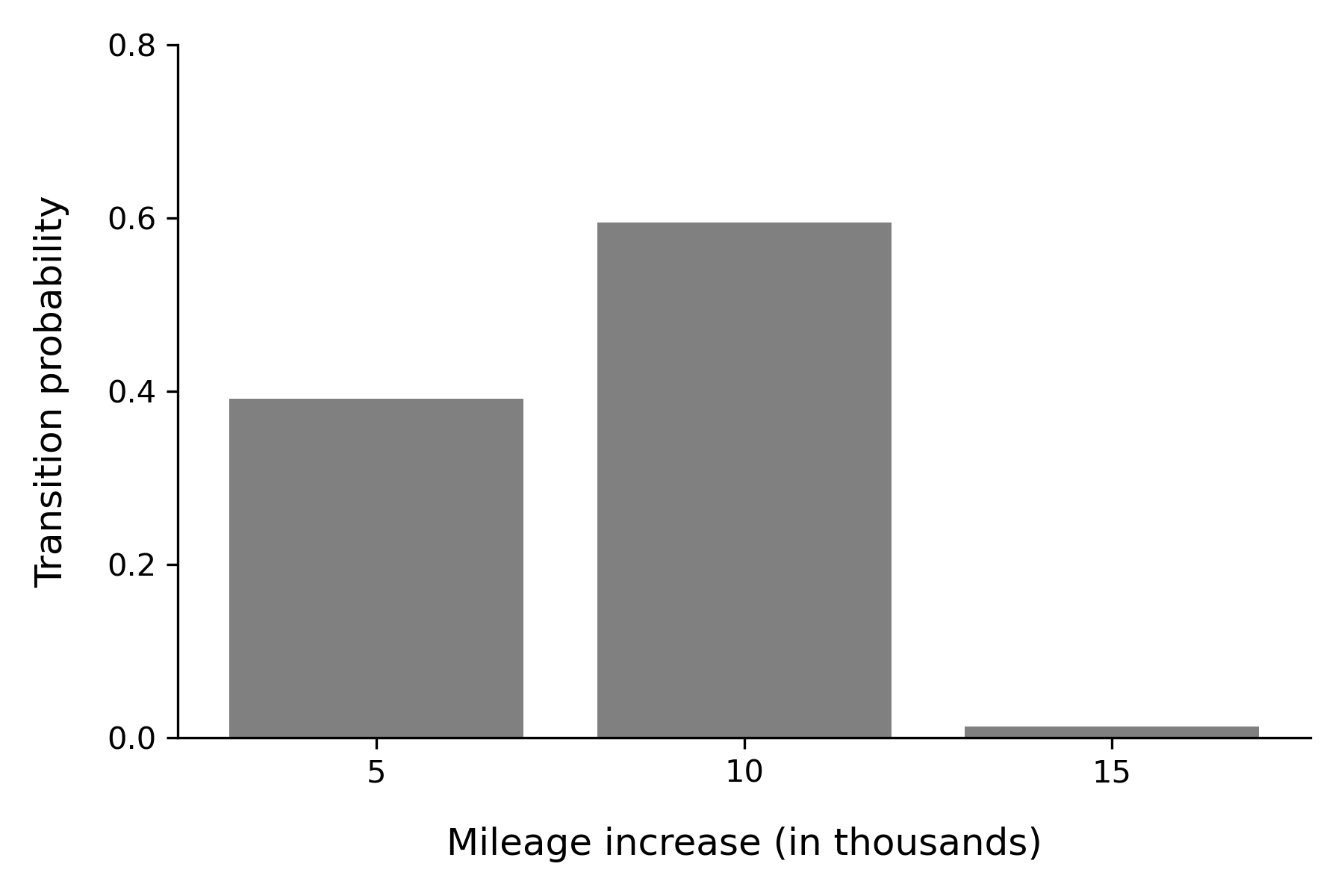}}
\caption{Estimated transition probabilities}\label{Estimated transition probabilities}
\end{figure}

The confidence level $\omega$ and the available number of observations $N_x$ determine the size of the ambiguity set as outlined in Equation (\ref{Mapping}). From now on, we mimic state-specific ambiguity sets by constructing them based on the average number of $55$ observations per state. Note that while the estimated distribution is the same for all mileage levels, its worst-case realization is not. However, there are only minor differences across mileage levels, so we focus our following discussion on a bus with an odometer reading of 75,000.\\

Figure \ref{Worst-case transition probabilities} shows the transition probabilities for different sizes of the ambiguity set. We vary the confidence level for the whole number of observations $(N_x = 55)$ on the left, while on the right, the level of confidence remains fixed $(\omega=0.95)$, and we cut the number of observations roughly in half. The larger the ambiguity set, the more probability is attached to higher mileage utilization, resulting in higher costs overall. For example, while the probability of mileage increases of 10,000 or more is an infrequent occurrence in the data, its probability increases first to 1.7\%. It then doubles to 2.5\% as we increase the confidence level. When only about half the data is available, this probability increases even further to 3.2\%.\\

\begin{figure}[h!]\centering
\subfloat[Variation in $\omega$, $(N_x = 55)$]{\scalebox{0.50}{\includegraphics{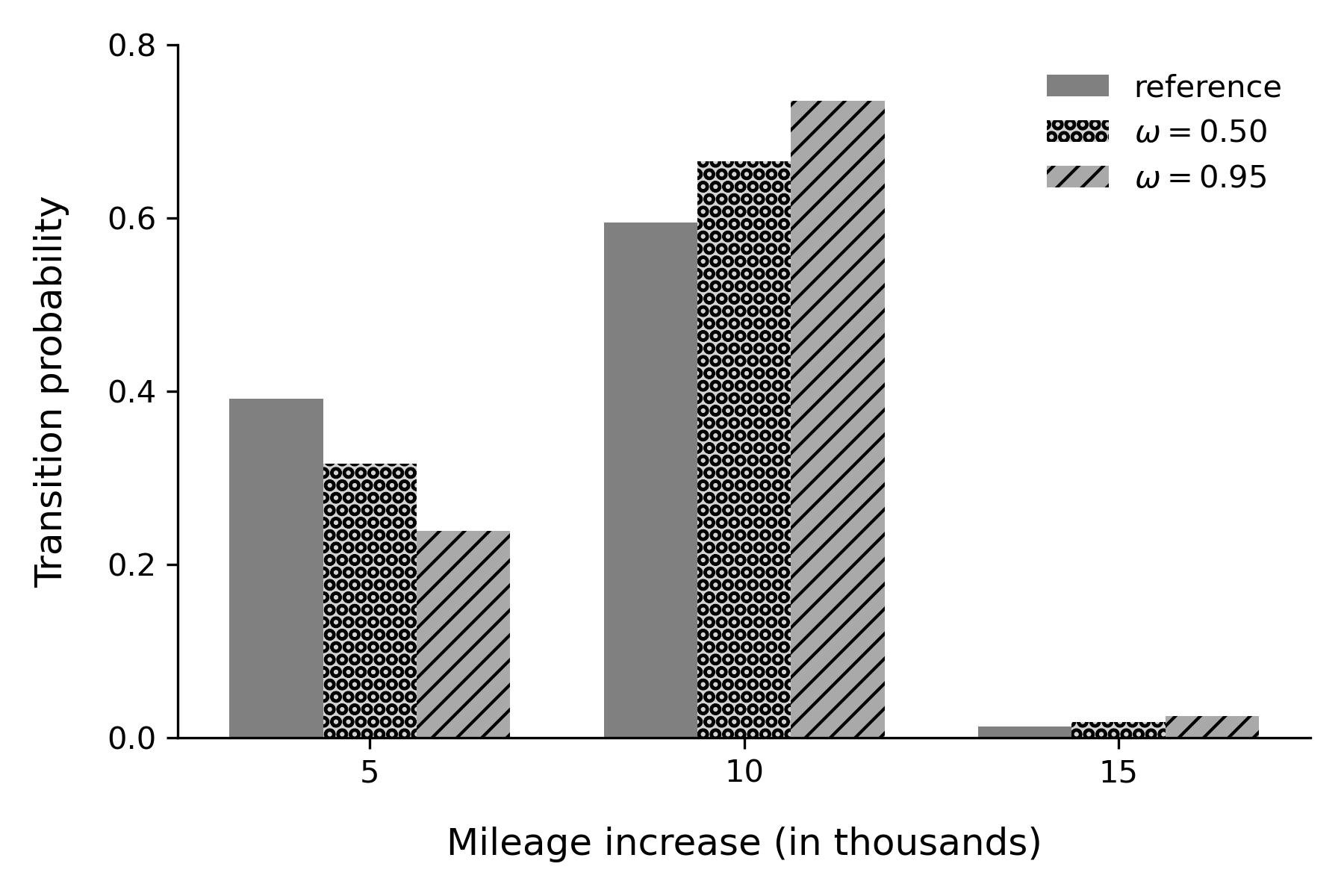}}}\hspace{0.3cm}
\subfloat[Variation in $N_x$, $(\omega = 0.95)$]{\scalebox{0.50}{\includegraphics{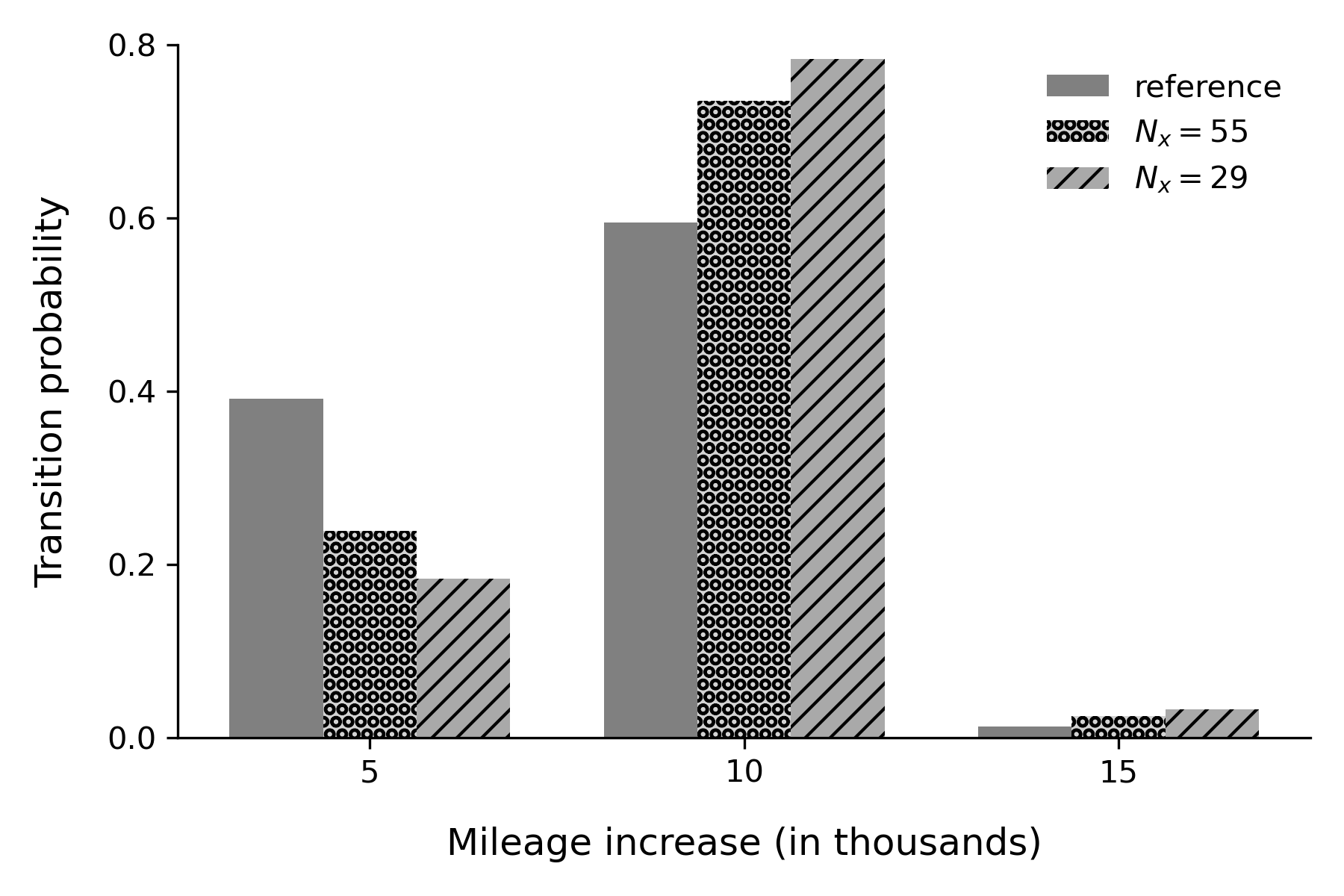}}}
\caption{Worst-case transition probabilities}\label{Worst-case transition probabilities}
\end{figure}

The decision-maker chooses whether to perform regular maintenance work on a bus or replace its complete engine each month. The assumed transition probabilities correspond to their worst-case transitions within the ambiguity set. As a result, any differences between the as-if and worst-case distributions translate into different maintenance decisions.\\

Figure \ref{Maintenance probabilities} shows the maintenance probabilities for different levels of accumulated mileage and alternative rules. Overall, the maintenance probability decreases with accumulated mileage, as maintenance becomes more costly than an engine replacement. Robust rules result in a higher probability of maintenance compared to the as-if decision rule. Under the worst-case transitions, a bus is more likely to experience higher usage during the period. As the cost of maintenance is determined by the mileage level at the beginning of the period, maintenance becomes more attractive. For example, again considering a bus with 75,000 miles, the as-if maintenance probability is 25\%, while it is 33\% $(\omega=0.50)$ and 43\% $(\omega=0.95)$ following the robust rule.\\

\begin{figure}[h!]\centering
\scalebox{0.75}{\includegraphics{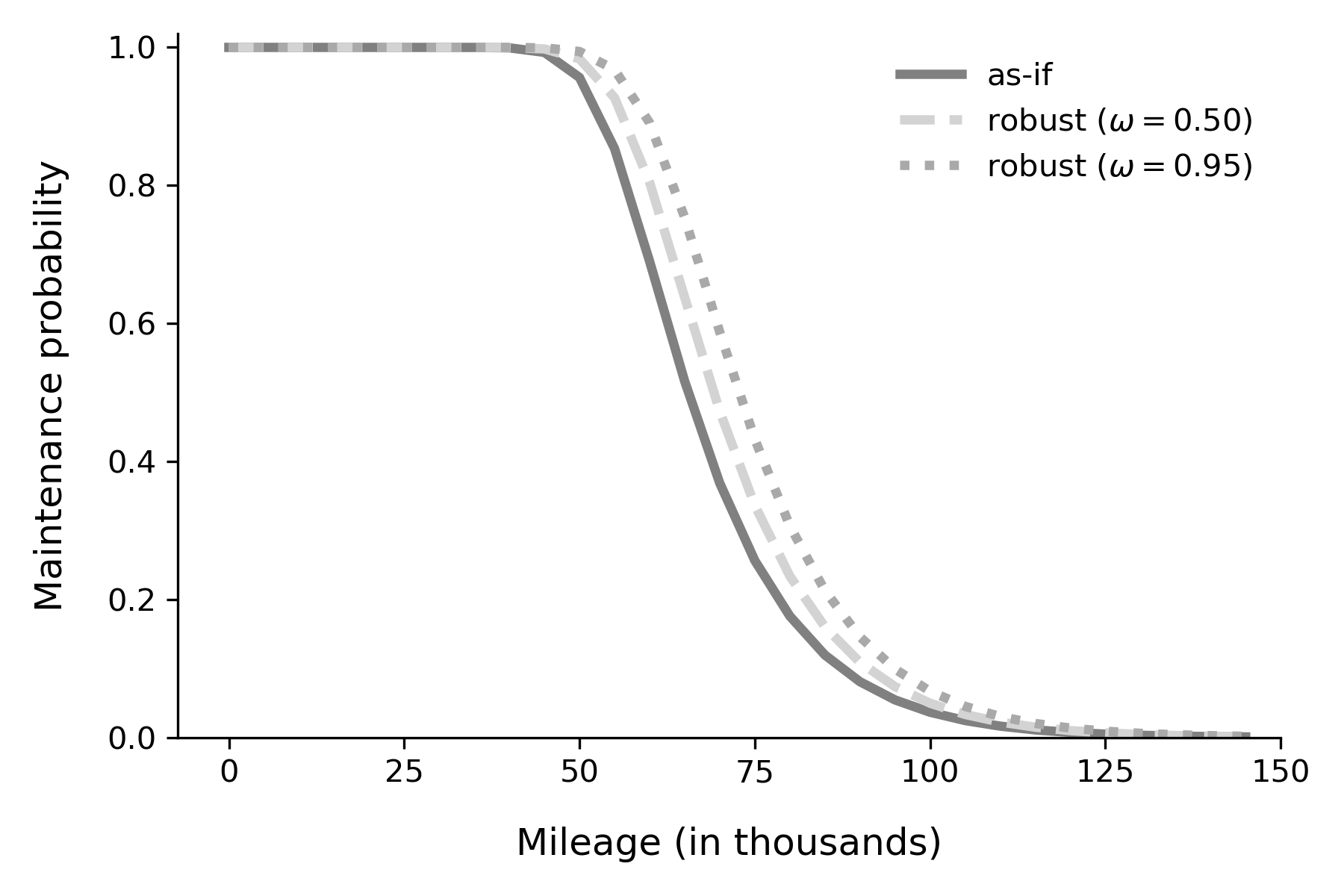}}
\caption{Maintenance probabilities}\label{Maintenance probabilities}
\end{figure}\FloatBarrier

To gain further insights into the differences between the as-if and robust decisions, we simulate a fleet of 1,000 buses for 100,000 months under the alternative decision rules.\\

Figure \ref{Single bus for alternative decision rules} shows the level of accumulated mileage over time for a single bus under different decision rules. It clarifies our simulation setup, where we apply different decision rules to the same bus. The realizations of observed transitions and unobserved signs of wear and tear remain the same. The bus accumulates more and more mileage until Harold Zurcher replaces the complete engine and the odometer is reset to zero. The first replacement happens after 20 months at 60,000 miles following the as-if decision rule, while it is delayed for another four months under the robust alternative $(\omega = 0.95)$. As its timing differs, the odometer readings will start to diverge after 20 months, even though monthly utilization remains the same.\\

\begin{figure}[h!]\centering
\scalebox{0.75}{\includegraphics{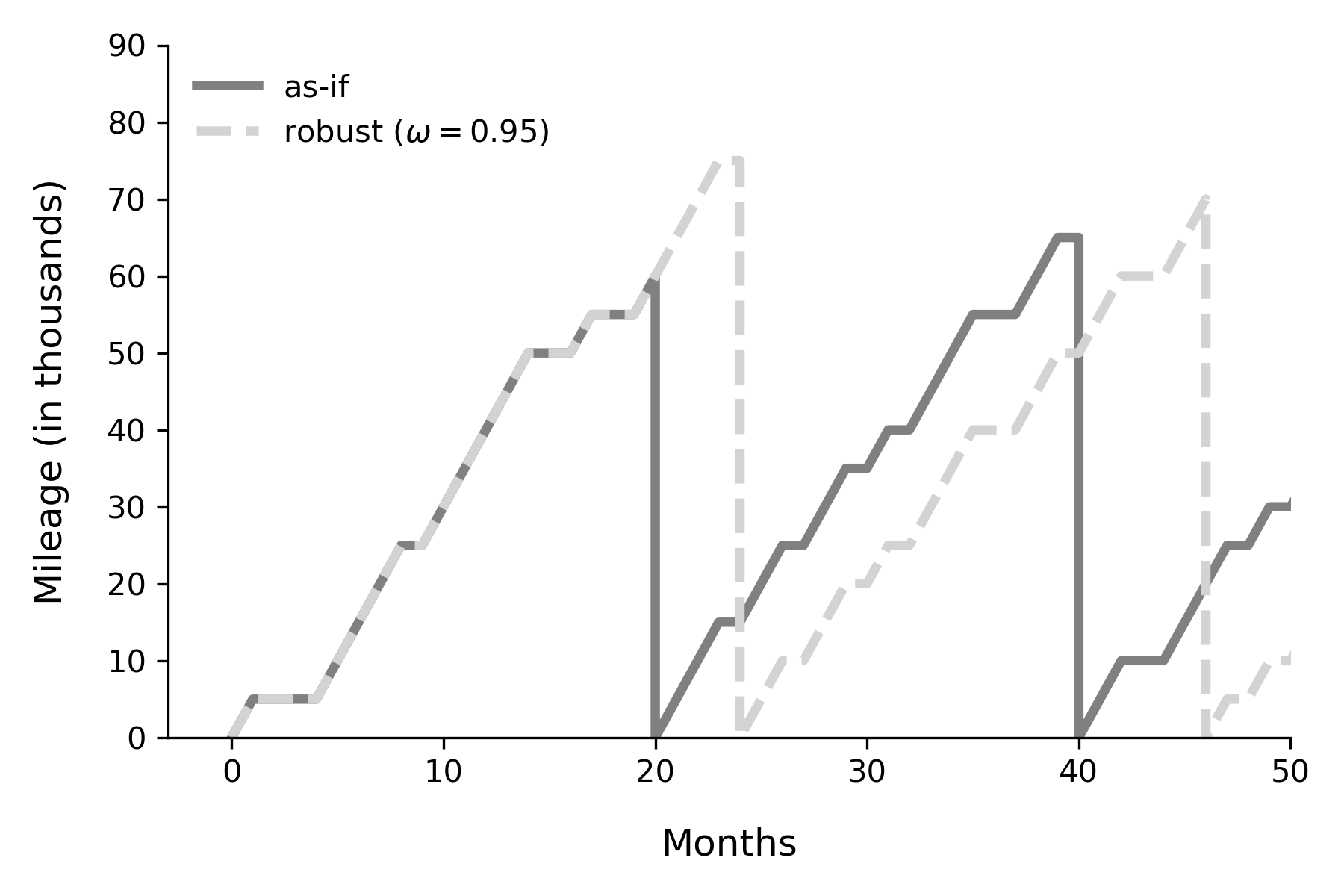}}
\caption{Single bus under alternative decision rules}\label{Single bus for alternative decision rules}
\end{figure}\FloatBarrier

We now evaluate the as-if and robust decisions at the boundary of the ambiguity set. We measure the performance of the alternative decision rules based on their total discounted utility under different assumed and actual mileage transitions.\\

Figure \ref{Performance of as-if decision rule} shows the performance of the as-if decision rule over time when the worst-case distribution for a confidence level of 0.95 governs the actual transitions. It illustrates the sensitivity of the as-if decision rule to perturbations in the transition probabilities. The solid line corresponds to its expected long-run performance without misspecification of the decision problem, while the dashed line indicates its observed performance. After about 20,000 months, it accumulates the expected long-run average cost and performs about 14\% worse overall.\\

\begin{figure}[h!]\centering
\scalebox{0.75}{\includegraphics{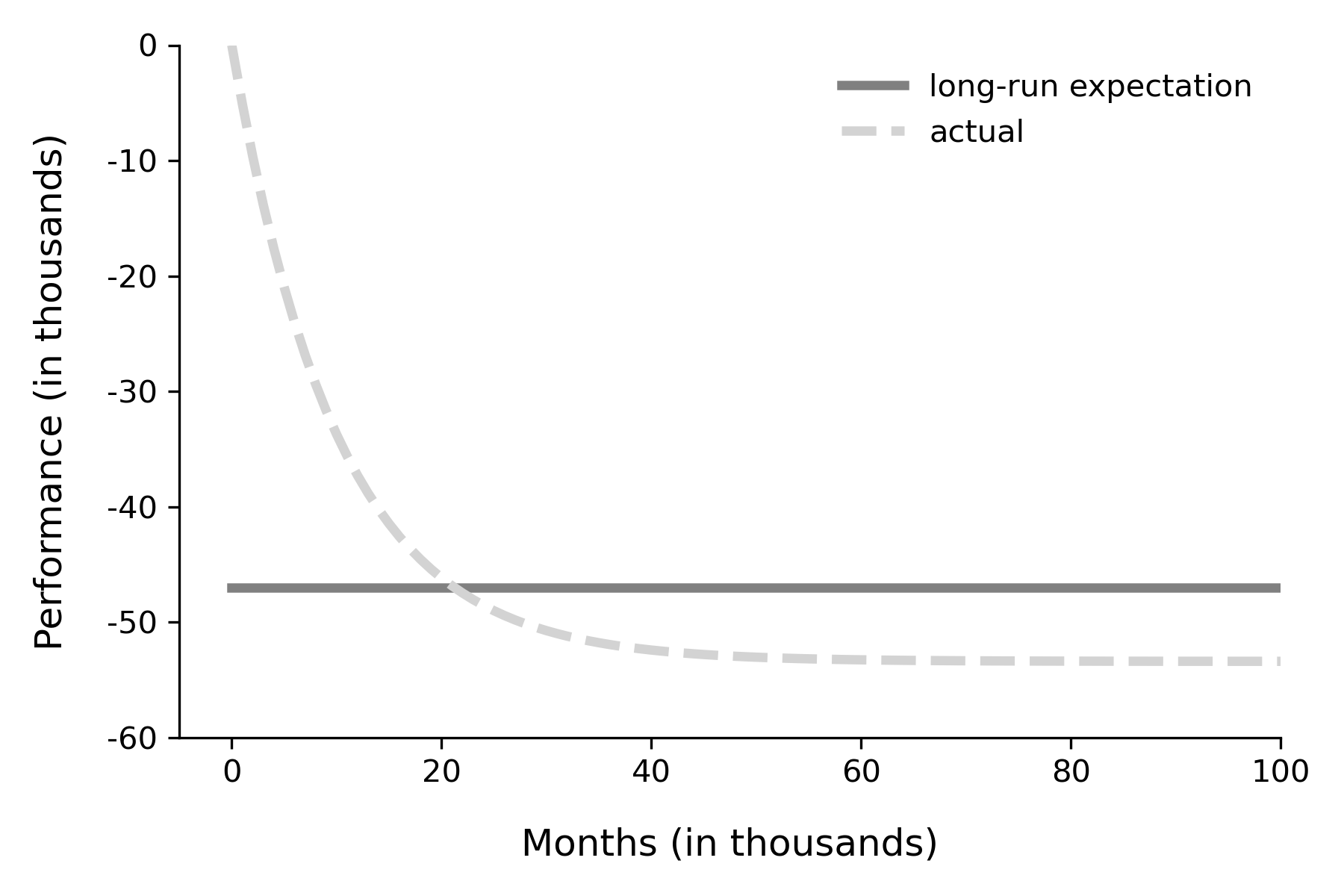}}
\caption{Performance of as-if decision rule}\label{Performance of as-if decision rule}
\end{figure}\FloatBarrier

Figure \ref{Performance and misspecification} shows the average difference in performance between the as-if and two robust decision rules with confidence levels of $0.50$ and $0.95$, respectively. The actual transitions follow the worst-case distribution with varying $\omega$. A positive value indicates that the robust decision rule outperforms the as-if decision rule. In the absence of any misspecification, the as-if decision rule must defeat any other decision rule. The same is true for the robust decision rule when the actual transitions are governed by the same $\omega$ used for their construction. Nevertheless, the as-if decision rule continues to outperform both robust decisions for moderate levels of $\omega$. For worst-case distributions with $\omega$ larger than 0.2, the first robust decision rule $(\omega=0.5)$ starts to beat the as-if decision rule. For the other robust decision rule $(\omega=0.95)$, the same is true for worst-case transitions of $\omega$ equal to 0.5.\\

\begin{figure}[h!]\centering
\scalebox{0.75}{\includegraphics{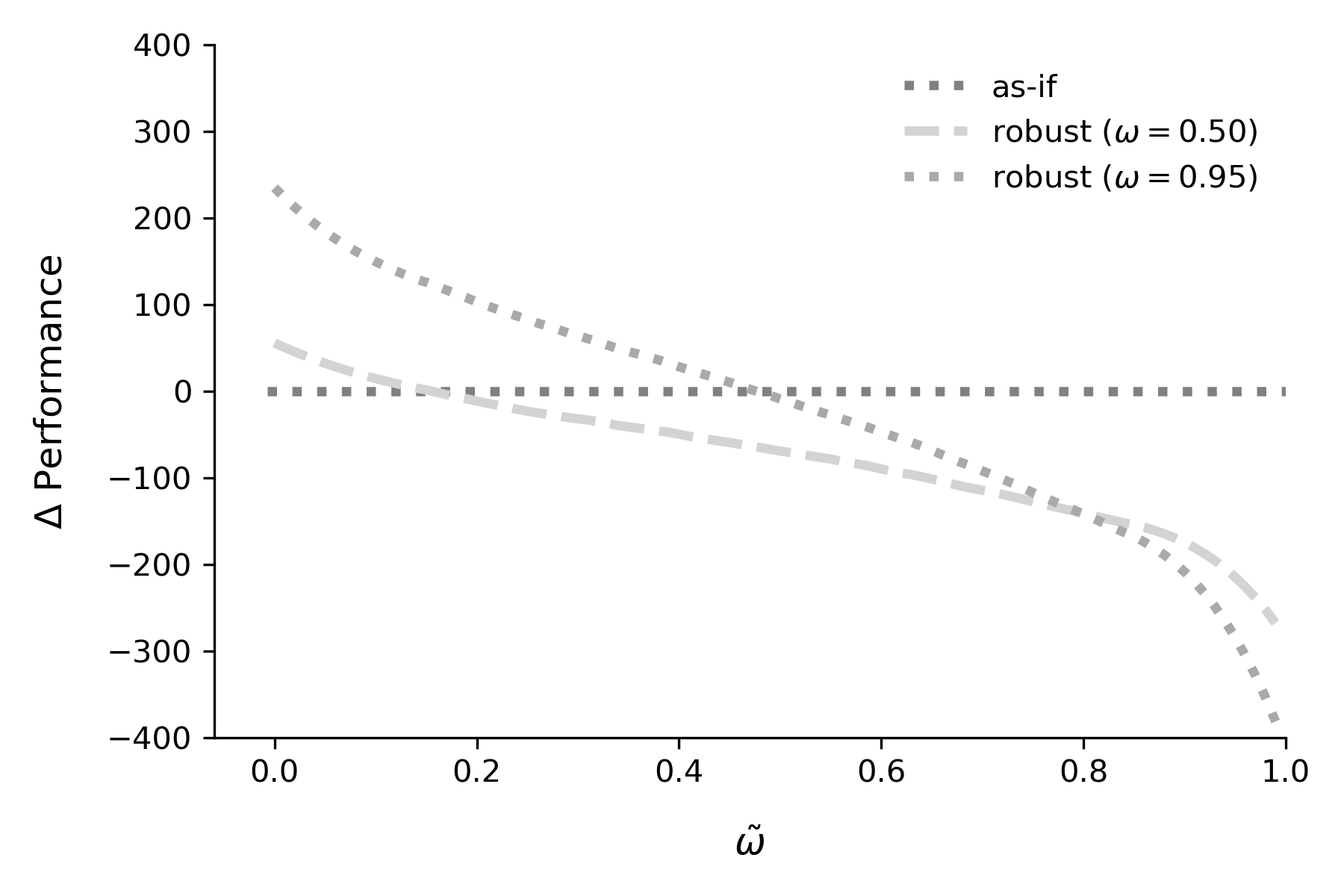}}
\begin{center}
\begin{minipage}[t]{0.80\columnwidth}\vspace{-0.75cm}
\item \scriptsize{\textbf{Notes:} We apply a Savitzky-Golay filter \citep{Savitzky.1964} to smooth the simulation results.}
\end{minipage}
\caption{Performance and misspecification}\label{Performance and misspecification}
\end{center}
\end{figure}\FloatBarrier

\subsection{Ex-ante analysis}
We now turn to the situation before any data are realized. We evaluate the ex-ante performance of as-if and robust decision functions over the whole probability simplex and determine the optimal level of robustness.\\

We operationalize our analysis as follows. In line with \citetalias{Rust.1987} assumption on the distribution of the mileage utilization, we specify a uniform grid with $0.1$ increments over the interior of the two-dimensional probability simplex $\mathring{\Delta}_3$. At each grid point, we draw 100 samples of 55 random mileage utilizations. For each sample, we solve several robust decision functions for a grid of $\omega = \{0.0, 0.1, \hdots, 1.0\}$ using the estimated transition probabilities. Note that the uncertainties are coupled across states, as the same underlying probability creates the sample of bus utilizations. Thus, the rectangularity assumption does not reflect the economic environment. However, we still impose it when constructing the robust decision functions to ensure tractability. We then simulate the implied decision rules' actual performance and compute their expected performance by averaging across the 100 runs for each grid point. Using this information, we measure the performance of the different decisions based on the maximin criterion, the minimax regret rule, and the subjective Bayes approach using a uniform prior.\\

In Figure \ref{Relative performance} we illustrate the differences in expected performance between a robust decision function $(\omega=0.1)$ and the as-if alternative over the probability simplex.
\begin{figure}[h!]\centering
\scalebox{0.75}{\includegraphics{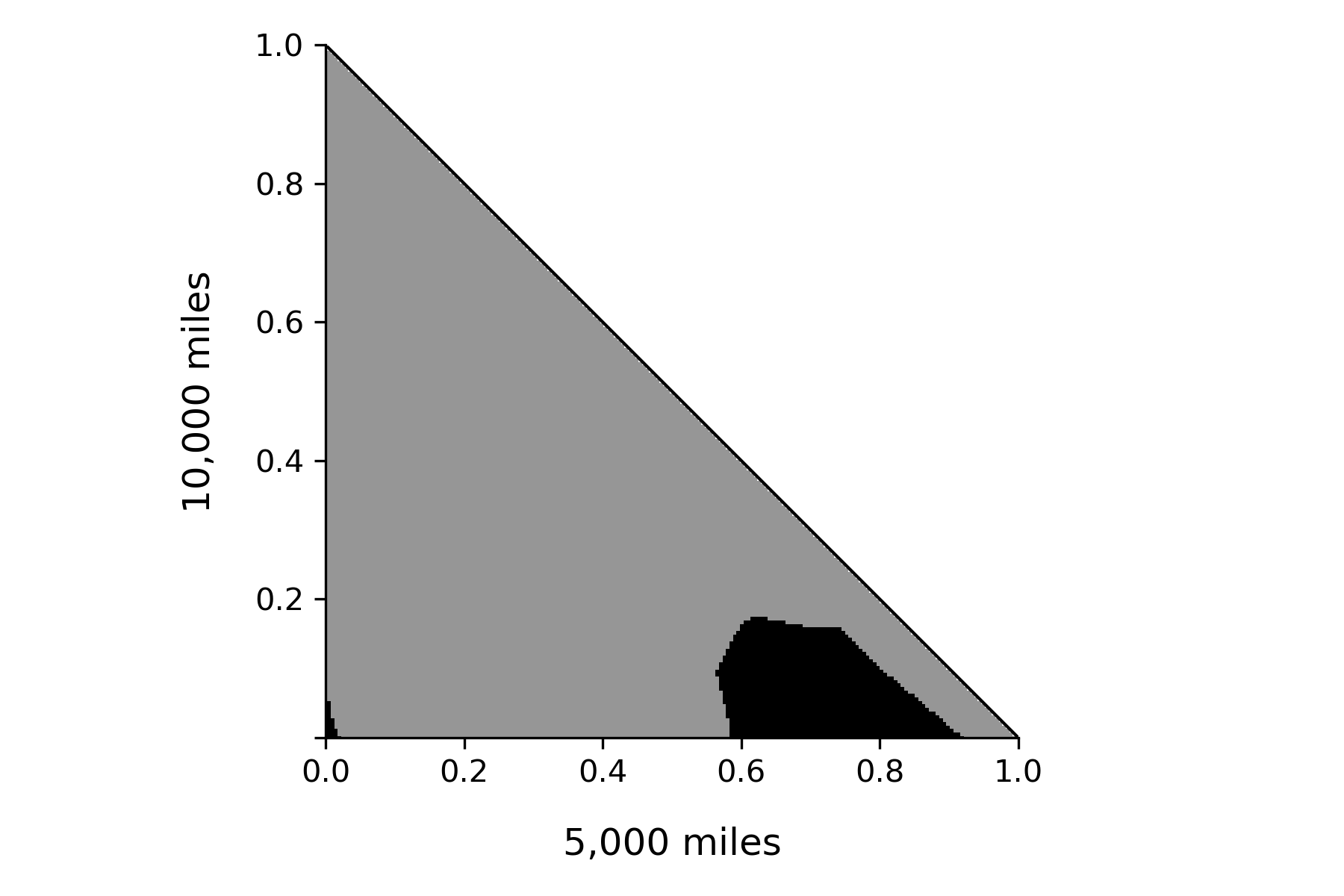}}
\caption{Relative performance of decision rules}\label{Relative performance}
\end{figure}\FloatBarrier
In the gray areas, the as-if decisions outperform the robust alternative based on their expected performance. The opposite is true for the black areas: robust decisions perform very well when the true probability of mileage increases of $5,000$ per month is high and when the true probability of increases amounting to $10,000$ is low. Otherwise, the as-if decisions outperform the robust alternative. Thus, no rule dominates the other, and it is essential to aggregate the performance over the whole probability simplex using decision theory before settling on a decision rule.\\

Figure \ref{Decision rules ranking} ranks the as-if decisions against selected robust alternatives for the different performance criteria.
\begin{figure}[h!]\centering
\scalebox{0.75}{\includegraphics{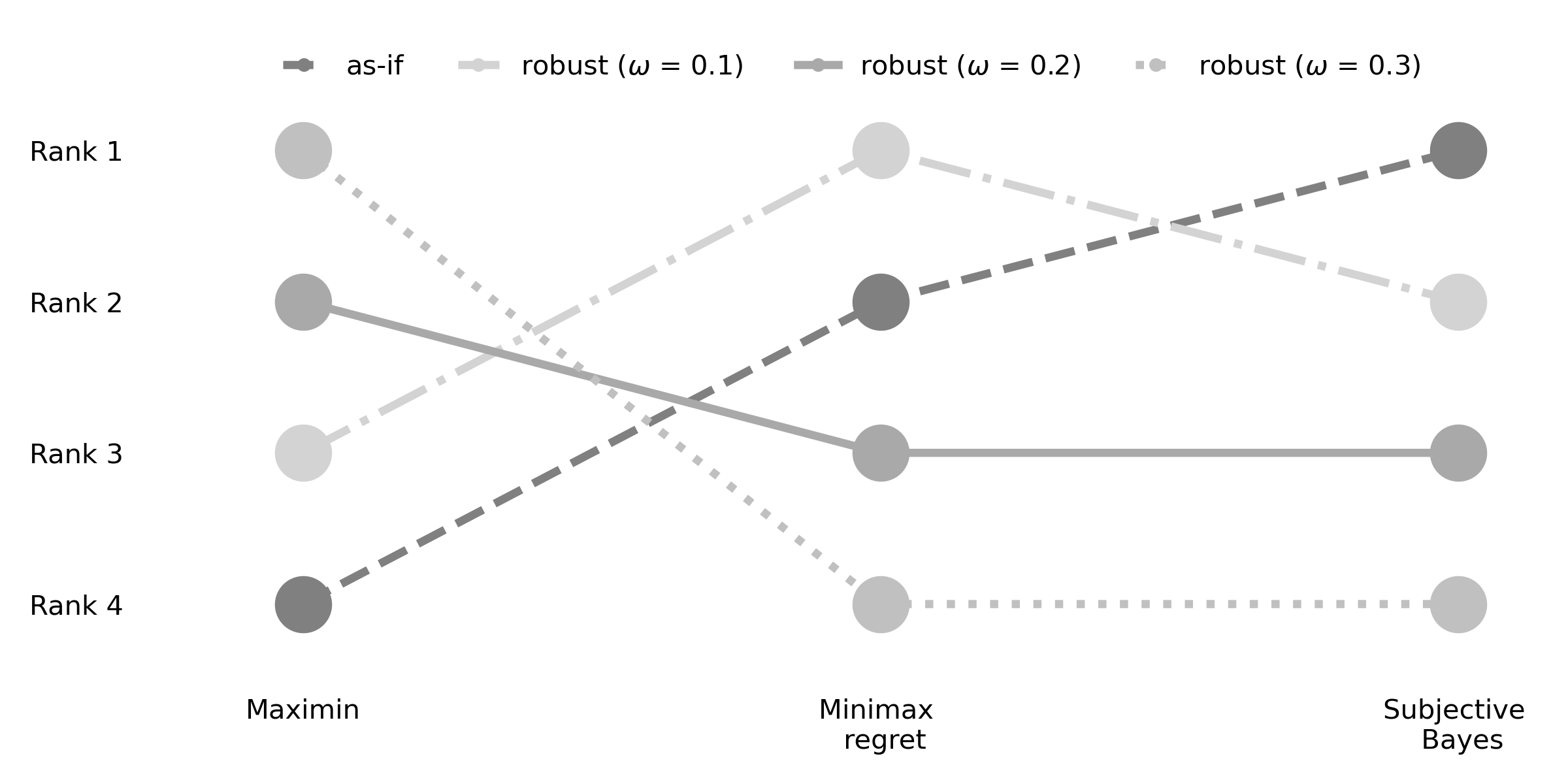}}
\caption{Ranking of decision rules}\label{Decision rules ranking}
\end{figure}\FloatBarrier
Based on a maximin criterion, decision functions rank higher when the confidence level $\omega$ used to construct them is greater. The decision function with $\omega = 0.3$ comes in first, while as-if decisions rank last. Thus, decision-makers can improve their worst-case outcomes by adopting a robust decision function. However, this comes at a cost, as indicated by the improved rankings for the as-if decision function as we move to different criteria. As-if decisions move to second place for minimax regret. The as-if decision rule comes in first when we aggregate performance across all states using a subjective Bayes approach with a uniform prior. Thus, our approach clarifies the trade-offs involved when choosing a particular decision function for decision-making.\\

We now determine the optimal size of the ambiguity set $\omega^*$ for each decision-theoretic criterium. Figure \ref{Optimality of decision rules} shows the minimum performance of the decision functions for varying levels of $\omega$ normalized between zero and one. Among all decision functions, robust decisions with $\omega=0.36$ have the highest minimum performance. They thus strike a balance between the conservatism of the worst-case approach and the protection against unfavorable transition probabilities. Based on the maximin criterion, the as-if decision function performs worst.

\begin{figure}[h!]\centering
\scalebox{0.75}{\includegraphics{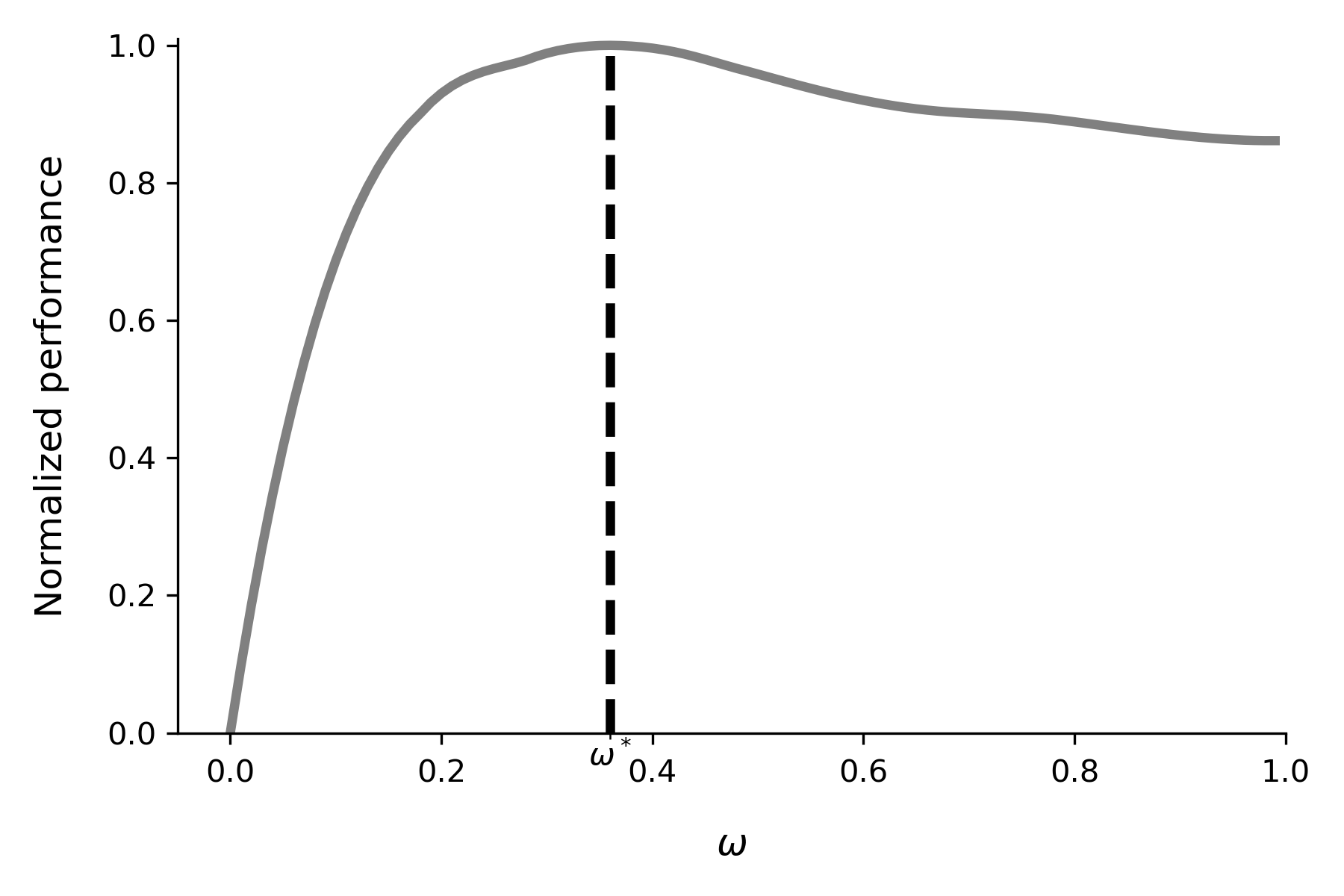}}
\caption{Optimality of decision rules}\label{Optimality of decision rules}
\end{figure}\FloatBarrier

The minimax regret criterion leads to a slightly reduced level of $\omega^*=0.1$. As-if decisions are optimal based on the subjective Bayes criterion with a uniform prior.

\FloatBarrier\section{Conclusion}\label{Conclusion}
Economists often estimate economic models on data and use the point estimates as a stand-in for the truth when studying the model's implications for optimal decision-making. This practice ignores model ambiguity, exposes the decision problem to misspecification, and ultimately leads to post-decision disappointment. We develop a framework to explore, evaluate, and optimize robust decision rules that explicitly account for the uncertainty in the estimation using statistical decision theory. We show how to operationalize our analysis by studying robust decisions in a stochastic dynamic investment model in which a decision-maker directly accounts for uncertainty in the model's transition dynamics.\\

As our core contribution, we combine ideas from data-driven robustness optimization \citep{Bertsimas.2018}, robust Markov decision processes \citep{Ben-Tal.2009}, and statistical decision theory \citep{Berger.2010} to optimize robustness in decision-making. This insight transfers directly to many other settings. For example, the COVID-19 pandemic provides a timely example of economists informing policy-making by using highly parameterized models in light of ubiquitous uncertainties \citep{Avery.2020}. When analyzing these models, economists treat many of their parameters as if they are known. However, their actual values are uncertain, as they are often estimated based on external data sources. Using statistical decision theory, our research illustrates how to conduct robust policy-making and to evaluate its relative performance against policies that ignore uncertainty. Such an approach promotes a sound decision-making process, as it provides decision-makers with the tools to systematically navigate the uncertainties they face \citep{Berger.2021}.

\newpage\bibliographystyle{apalike}

\newpage
\begin{appendices}\setcounter{page}{1}
\renewcommand{\thepage}{\thesection-\arabic{page}}
\newpage
\renewcommand*{\thepage}{A-\arabic{page}}
\begin{appendices}

\section{The robust contraction mapping}\label{The robust contraction mapping}
\cite{Rust.1987} shows that the expectation of the next period's value function is a fixed point on the mileage states $x$ only. He uses the regenerative property of the mileage process and introduces a separate notion $\tilde{EV}(x)$ for the expected value function of maintenance. $\tilde{EV}(x)$ is the fixed point of the contraction mapping defined as follows: For all $v \in \mathbb{V}$

\begin{equation}
  \tilde{\Lambda}(v)(x) = \sum_{x' \in X} p(x'|x) \log \sum_{a \in \{0, 1\}} \exp\bigl( u(x' , a) + \delta \thin v((1 - a) \thin x') \bigr), \, x\in X.
\end{equation}

Following \cite{Iyengar.2005}, we adopt a similar approach and show:

\begin{Theorem}\label{robust_operator_rust_thm} Let the robust Bellman operator $\Lambda:\mathbb{V}\rightarrow\mathbb{V}$ be defined as follows: For all $v \in \mathbb{V}$
\begin{align}
  \Lambda(v)(x) &= \int \max_{a \in \{0,1\}} \biggl[u(x, a) + \epsilon(a) + \delta\thin \underset{p \in \mathcal{P}((1 - a) \thin x, \thin \omega)}{\min} \sum_{x^\prime \in X} p(x^\prime) v(x^\prime) \biggr] q(d\epsilon)\nonumber \\
  \label{robust_zurcher_operator}
  &= \log \sum_{a \in \{0, 1\}} \exp \biggl[ u(x , a) + \delta \underset{p \in \mathcal{P}((1 - a) \thin x, \thin \omega)}{\min} \sum_{x^\prime \in X} p(x^\prime) v(x^\prime) \biggr].
\end{align}

Then $\Lambda(\cdot)$ is a contraction mapping on $\big(\thin\mathbb{V},\thin \left\| \cdot \right\|_\infty\big)$ with unique fixed point $EV$. \\
\end{Theorem}
\begin{proof}
Let $v, w \in \mathbb{V}$ be arbitrary. Fix $x \in X$ and assume without loss of generality that $\Lambda(w)(x) \geq \Lambda(v)(x)$. Let $\nu >0$ be arbitrary. Then choose $p^a \in \mathcal{P}((1 - a) x, \omega)$, such that
\begin{align*}
&\max_{a\in\{0, 1\}}[u(x, a) + \epsilon(a) + \delta \underset{p \in \mathcal{P}((1 - a) \thin x, \thin \omega)}{\min} \sum_{x^\prime \in X} p(x^\prime) v(x^\prime)] \geq \\  &\max_{a\in\{0, 1\}}[u(x, a) + \epsilon(a) + \delta \sum_{x^\prime \in X} p^a(x^\prime) v(x^\prime)] - \nu.
\end{align*}
By construction:
\begin{align*}
&\max_{a\in\{0, 1\}}[u(x, a) + \epsilon(a) + \delta \underset{p \in \mathcal{P}((1 - a) \thin x, \thin \omega)}{\min} \sum_{x^\prime \in X} p(x^\prime) w(x^\prime)] \leq \\  &\max_{a\in\{0, 1\}}[u(x, a) + \epsilon(a) + \delta \sum_{x^\prime \in X} p^a(x^\prime) w(x^\prime)].
\end{align*}
\cite{Rust.1988} shows for any conditional distribution measure $p$ and mileage state $x\in X$:
\begin{align*}
&\max_{a\in\{0, 1\}}[u(x, a) + \epsilon(a) + \delta \sum_{x^\prime \in X} p(x^\prime) w(x^\prime)] - \max_{a\in\{0, 1\}}[u(x, a) + \epsilon(a) + \delta \sum_{x^\prime \in X} p(x^\prime) v(x^\prime)] \\
\leq & \thin\delta \max_{a\in\{0, 1\}} |\sum_{x^\prime \in X} p(x^\prime) (w(x^\prime) - v(x^\prime)| \leq \delta \thin \left\|w - v\right\|_\infty.
\end{align*}
This holds in particular for $p^a$, which yields:
\begin{align*}
  0 &\leq \Lambda(w)(x) - \Lambda(v)(x) \\
  &\leq \int \biggl(\max_{a\in\{0, 1\}}[u(x, a) + \epsilon(a) + \delta \sum_{x^\prime \in X} p^a(x^\prime) w(x^\prime)] \\
  &\qquad\quad - \max_{a\in\{0, 1\}}[u(x, a) + \epsilon(a) + \delta \sum_{x^\prime \in X} p^a(x^\prime) v(x^\prime)] + \nu \biggr) q(d\epsilon)\\
  &\leq \int (\delta \thin\left\|w - v\right\|_\infty + \nu)\thin q(\epsilon)\\
  &= \delta \thin \left\|w - v\right\|_\infty + \nu.
\end{align*}

Arguing vice versa for $\Lambda(w)(x) \leq \Lambda(v)(x)$, this implies that
\begin{equation*}
\left\| \Lambda(w) - \Lambda(v)\right\|_\infty \leq \delta \left\|w - v\right\|_\infty + \nu.
\end{equation*}
With $\nu$ arbitrary and $\delta \in [0,\thin 1 )$ this shows that $\Lambda$ is a contraction mapping on $\mathbb{V}$ with respect to $\left\|\cdot\right\|_\infty$. As $\big(\thin\mathbb{V},\thin \left\| \cdot \right\|_\infty\big)$ is a Banach space, the result is established.
\end{proof}

\end{appendices}

\newpage\FloatBarrier\end{appendices}

\end{document}